\documentclass{article}

\usepackage[english]{babel}
\usepackage{amsthm,amsmath,lipsum}
\usepackage{amssymb}
\usepackage{float}
\usepackage[ruled,vlined]{algorithm2e}
\usepackage{xcolor}
\usepackage{fullpage}
\usepackage{amssymb}
\usepackage{centernot}
\usepackage{authblk}

\usepackage{tikz}
\usetikzlibrary{decorations.pathreplacing,calc}
\newcommand{\tikzmark}[1]{\tikz[overlay,remember picture] \node (#1) {};}

\newcommand*{\AddNote}[4]{%
    \begin{tikzpicture}[overlay, remember picture]
        \draw [decoration={brace,amplitude=0.5em},decorate,ultra thick,blue]
            ($(#3)!(#1.north)!($(#3)-(0,1)$)$) --  
            ($(#3)!(#2.south)!($(#3)-(0,1)$)$)
                node [align=center, text width=2.5cm, pos=0.5, anchor=west] {#4};
    \end{tikzpicture}
}%

\usepackage{soul}
\usepackage{natbib}
\usepackage{comment}
\usepackage{caption}
\usepackage{subcaption}

\usepackage{thm-restate}

\usepackage[colorlinks=True,linkcolor=black,anchorcolor=black,citecolor=black,filecolor=black,menucolor=black,runcolor=black,urlcolor=black]{hyperref}
\usepackage[capitalise,nameinlink
    ,noabbrev]{cleveref}

\newcommand{\bstf}{EFX-best}
\newcommand{\efxf}{EFX-feasible}
\newcommand{\eefxf}{EEFX-feasible}
\newcommand{\mxsf}{MXS-feasible}

\newcommand{\eflf}{EFL-feasible}
\newcommand{\eflmf}{MXS+EFL-feasible}

\newcommand{\efxenvy}{EFX-envy}
\newcommand{\efxenvies}{EFX-envies}
\newcommand{\eflenvy}{EFL-envy}
\newcommand{\eflenvies}{EFL-envies}

\newcommand{\Qa}{X_p}
\newcommand{\Qb}{X_q}
\newcommand{\Qc}{X_r}
\newcommand{\Ip}{p}
\newcommand{\Iq}{q}
\newcommand{\Ir}{r}

\newcommand{\ggeq}{\geq\!\!\!\geq}
\renewcommand{\gg}{>\!\!\!>}

\newcommand{\PA}{X_p}  
\newcommand{\QB}{X_q}  
\newcommand{\RE}{R}  
\newcommand{\SC}{X_r}
\newcommand{\PAp}{X'_{p'}}  
\newcommand{\QBp}{X'_{q'}}
\newcommand{\REp}{R'}  

\newcommand{\ck}{C_k}

\newcommand{\ckx}{C_k(\X,f)}
\newcommand{\ckxp}{C_k(\xp,f')}
\newcommand{\ckxz}{C_k(\xz,f'')}
\newcommand{\sckx}{SC_k(\X,f)}

\newcommand{\TT}{T}

\newcommand{\Za}{Z_a}
\newcommand{\Zb}{Z_b}
\newcommand{\Zap}{Z'_a}
\newcommand{\Zbp}{Z'_b}

\newcommand{\inva}{invariant A}
\newcommand{\invb}{invariant B}

\newcommand{\range}{I}

\newcommand{\xii}{x_p^i}
\newcommand{\xj}{x_p^j}
\newcommand{\xu}{x_p^u}

\newcommand{\X}{\mathbf{X}}
\newcommand{\Y}{\mathbf{Y}}
\newcommand{\Z}{\mathbf{Z}}

\newcommand{\xp}{\mathbf{X'}}
\newcommand{\yp}{\mathbf{Y'}}
\newcommand{\zp}{\mathbf{Z'}}
\newcommand{\xz}{\mathbf{X^{\prime \prime}}}

\newcommand{\efxfset}{EFXF}
\newcommand{\bstfset}{EFXBest}
\newcommand{\eflfset}{EFLF}
\newcommand{\eflmfset}{MXS+EFLF}
\newcommand{\eefxfset}{EEFXF}
\newcommand{\mxsfset}{MXS}

\newcommand{\Isnotefxf}[3]{$#1 \not\in \text{\efxfset}_{#2}(#3)$}

\newcommand{\Isnotbstf}[3]{$#1 \not \in \text{\bstfset}_{#2}(#3)$}
\newcommand{\Isnoteefxf}[3]{$#1 \not\in \text{\eefxfset}_{#2}(#3)$}
\newcommand{\Isnoteflmf}[3]{$#1 \not\in \text{\eflmfset}_{#2}(#3)$}

\newcommand{\iseflf}[3]{$#1 \in \text{\eflfset}_{#2}(#3)$}

\newcommand{\isefxffor}[3]{$#1 \in \text{\efxfset}_{#2}(#3)$}
\newcommand{\isnotefxffor}[3]{$#1 \not\in \text{\efxfset}_{#2}(#3)$}
\newcommand{\isbstffor}[3]{$#1 \in \text{\bstfset}_{#2}(#3)$}

\newcommand{\iseflmffor}[3]{$#1 \in \text{\eflmfset}_{#2}(#3)$}
\newcommand{\iseflffor}[3]{$#1 \in \text{\eflfset}_{#2}(#3)$}
\newcommand{\iseefxffor}[3]{$#1 \in \text{\eefxfset}_{#2}(#3)$}
\newcommand{\ismxsffor}[3]{$#1 \in \text{\mxsfset}_{#2}(#3)$}

\newcommand{\isnoteefxffor}[3]{$#1 \not\in \text{\eefxfset}_{#2}(#3)$}

\newcommand{\goodcancelable}{good cancelable}

\newcommand{\with}{relative to}

\newcommand{\AssociationFunction}{Association Function}

\newcommand{\mef}{MXS+EFL}
\newcommand{\mefaf}{\mef\ association function}
\newcommand{\fmefaf}{full \mef\ association function}
\newcommand{\gegraph}{envy graph}

\newcommand{\empt}{0}
\newcommand{\f}[1]{f(#1)}

\newcommand{\argmax}{\textsc{Best}}
\newcommand{\argmin}{\textsc{Worst}}

\newcommand{\lowerval}[1]{\prec_{#1}}
\newcommand{\lowereqval}[1]{\preceq_{#1}}
\newcommand{\greaterval}[1]{\succ_{#1}}
\newcommand{\greatereqval}[1]{\succeq_{#1}}

\newcommand{\CaseA}{Case 1}
\newcommand{\CaseB}{Case 2}
\newcommand{\CaseC}{Case 3}
\newcommand{\CaseD}{Case 4}

\SetCommentSty{\mycommfont}
\LinesNumbered

\theoremstyle{definition}
\newtheorem{definition}{Definition}[section]
\theoremstyle{remark}
\newtheorem*{remark}{Remark}
\theoremstyle{lemma}
\newtheorem{lemma}{Lemma}[section]
\theoremstyle{theorem}
\newtheorem{theorem}{Theorem}[section]
\theoremstyle{corollary}
\newtheorem{corollary}{Corollary}[section]
\theoremstyle{observation}
\newtheorem{observation}{Observation}[section]

\title{Simultaneously Satisfying MXS and EFL\thanks{Vasilis Gkatzelis was partially supported by NSF CAREER award CCF-2047907.}}
\author[1]{Arash Ashuri\thanks{arash.ashoori0330@sharif.edu}}
\author[2]{Vasilis Gkatzelis\thanks{gkatz@drexel.edu}}
\affil[1]{Sharif University of Technology}
\affil[2]{Drexel University}
\date{}

\begin{document}
\maketitle

\begin{abstract}
    The two standard fairness notions in the resource allocation literature are proportionality and envy-freeness. If there are $n$ agents competing for the available resources, then proportionality requires that each agent receives at least a $1/n$ fraction of their total value for the set of resources. On the other hand, envy-freeness requires that each agent weakly prefers the resources allocated to them over those allocated to any other agent. Each of these notions has its own benefits, but it is well known that neither one of the two is always achievable when the resources being allocated are indivisible. As a result, a lot of work has focused on satisfying fairness notions that relax either proportionality or envy-freeness. 
    
    In this paper, we focus on MXS (a relaxation of proportionality) and EFL (a relaxation of envy-freeness). Each of these notions was previously shown to be achievable on its own~\citep{BBMN18,CGRSV22}, and our main result is an algorithm that computes allocations that simultaneously satisfy both, combining the benefits of approximate proportionality and approximate envy-freeness. In fact, we prove this for any instance involving agents with valuation functions that are restricted MMS-feasible, which are more general than additive valuations. Also, since every EFL allocation directly satisfies other well-studied fairness notions like EF1, $\frac{1}{2}$-EFX, $\frac{1}{2}$-GMMS, and $\frac{2}{3}$-PMMS, and every MXS allocation satisfies $\frac{4}{7}$-MMS, the allocations returned by our algorithm simultaneously satisfy a wide variety of fairness notions and are, therefore, universally fair~\citep{AMN20}.
\end{abstract}

\section{Introduction}

    The two most natural and well-studied notions of fairness in the resource allocation literature are proportionality and envy-freeness. Proportionality determines whether an allocation is fair using an \emph{absolute value benchmark} for each agent: if there are $n$ agents competing for the resources, each of them should receive a bundle they value at least $1/n$ times their total value for all the resources. The only consideration is whether an agent receives this ``fair share,'' irrespective of what other agents get. On the other hand, envy-freeness uses a \emph{relative value benchmark} that  depends on what other agents get: an allocation is envy-free if no agent prefers another agent's bundle over their own. Unfortunately, when the resources are indivisible, i.e., they cannot be shared between agents, it is well-known that both these notions may be unachievable. Hence, a lot of recent work on resource allocation has instead focused on relaxations of these two notions.
    
    The first relaxation of envy-freeness to be considered was {\em envy freeness up to some good} (EF1); it was formally defined by \citet{B10} but shown to be achievable earlier on, by \citet{LMMS04}. An allocation is EF1 if for any two agents $i,j$ with assigned bundles $X, Y$, respectively, there exists some good $g \in Y$ such that $i$ weakly prefers $X$ to $Y\setminus\{g\}$. 
    In this work, we focus on {\em envy-freeness up to one less-preferred good} (EFL), which is strictly more demanding than EF1. It was introduced by \citet{BBMN18}, who proved the existence of EFL allocations when the agents' valuations over resources are additive. An allocation is EFL if for any two agents $i,j$ with assigned bundles $X, Y$, respectively, either $|Y|\leq 1$ or there exists some good $g\in Y$ such that $i$ weakly prefers $X$ not only to $Y\setminus\{g\}$ but also to $\{g\}$.
    An even more demanding notion is {\em envy freeness up to any good} (EFX), introduced by \citet{CKMPSW19}. An allocation is EFX if for any two agents $i,j$ with assigned bundles $X, Y$, respectively, \emph{for every} good $g\in Y$ agent $i$ weakly prefers $X$ to $Y\setminus \{g\}$. However, despite significant effort, the existence of EFX allocations remains unknown, even for instances with just four agents.
    
    In terms of relaxations of proportionality, the most well-studied notion is the \emph{maximin fair share} (MMS), introduced by \cite{B10}, 
    but it is not always achievable \citep{KPW16} . 
    Inspired by MMS and EFX, and aiming to sidestep their existence issues, \citet{CGRSV22} introduced two interesting and achievable fairness notions: \emph{epistemic EFX} (EEFX) and \emph{minimum EFX share} (MXS).
    In this work, we focus on MXS: an allocation $\X$ is MXS if for every agent $i$ there exists an allocation $\Y$ such that $i$ weakly prefers $X_i$ (their bundle in $\X$) to $Y_i$ (their bundle in $\Y$), 
    and for every other bundle $Y_j\in \Y$, $i$ weakly prefers $Y_i$ to $Y_j\setminus \{g\}$ \emph{for every} good $g\in Y_j$. In other words, $\Y$ satisfies the EFX condition from $i$'s perspective, and $i$ is weakly happier in $\X$ compared to $\Y$. This notion sets an absolute value benchmark for each agent $i$, which is equal to $i$'s value in the worst allocation (from $i$'s perspective) that still satisfies $i$'s EFX condition.
    
    One of the open problems proposed by \citet{CGRSV22} was to study whether their new fairness notions can be combined with other fairness notions, like EF1. Our main result shows that MXS can be combined not just with EF1, but with EFL.

\paragraph{Our Results}
    Our main result is a constructive argument showing that there always exist allocations that are simultaneously EFL and MXS, thus simultaneously satisfying both relative and absolute fairness benchmarks. In fact, this argument works for the large class of \emph{restricted MMS-feasible valuations}, which generalizes the well-studied classes of additive, budget-additive, and unit-demand valuations. Our algorithm gradually constructs this allocation by carefully searching through the space of partitions guided by the agent preferences. From a technical standpoint, the main obstacle is to prove that this algorithm always terminates and does not end up stuck in an infinite loop. To achieve that, we use non-trivial invariants and potential functions.
    
    This result also has other implications. Specifically, for additive valuations, the EFL and MXS notions directly imply other known approximate fairness notions: every MXS allocation is $\frac{4}{7}$-MMS, and every EFL allocation is EF1, $\frac{1}{2}$-GMMS, $\frac{1}{2}$-EFX and $\frac{2}{3}$-PMMS (we discuss and prove these implications in \cref{sec:implications}).  Therefore, allocations satisfying both MXS and EFL simultaneously provide all the aforementioned fairness guarantees, leading to more improvements over prior work. For example, to the best of our knowledge, the best previously known combination of EFL and $\alpha$-MMS achieves $\alpha=\frac{1}{2}$~\cite{BBMN18}, while our result implies $\alpha=\frac{4}{7}$ is achievable.
   
\subsection*{Other Related Work}
        
    EFX allocations have been shown to exist in restricted settings, such as when  agents' valuations 
    are identical \citep{PR20}, additive but only up to three types \citep{PGNV24}. or binary \citep{HPPS20} (subsequently extended to  bi-valued \citep{ABRHV21}). For three-agent instances, EFX allocations were shown to exist for additive valuations \citep{CGM24}, subsequently extended to nice cancelable valuations \citep{BCFF21} and MMS-feasible valuations \citep{ACGMM22}. 
    For four-agent instances, EF2X allocations were shown to exist for cancelable valuations \citep{AGS24}.

    \noindent Another line of research has aimed to achieve multiplicative approximations of EFX. \citet{PR20} showed the existence of $1/2$-EFX allocations for subadditive valuations (which was subsequently derived in poly-time by \cite{CCLW19}), and \citet{AMN20} proved the existence of $1/\phi\approx 0.618$-EFX allocations for additive valuations. 
        For more restrictive valuation functions, \citet{MS23} and \citet{ARS24} proved the existence of $2/3$-EFX allocations. 
    \citet{BKP24} gave also improvements on the EFX approximation for restricted settings. 
    \citet{FHLSY21} proved the existence of $0.73$-EFR allocations, in which EFR is a relaxation of EFX.

    Also, a lot of work has focused on relaxations of EFX. One such relaxation is to donate some of the goods and achieve EFX with a ``partial allocation'' of the rest. This was first studied by \citet{CGH19} who showed the existence of a partial allocation that satisfies EFX and its Nash social welfare is half of the maximum possible. \citet{CKMS20} proved that EFX allocations exist for $n$ agents if up to $n-1$ goods can be donated to a ``charity" and, moreover, no agent envies the charity. \citet{BCFF21} improved this number to $n-2$ agents with nice cancelable valuations, and \citep{M24} extend this to monotone valuations.
    \citet{BCFF21} further showed that an EFX allocation exists for $4$ agents with at most one unallocated good.   
    The number of unallocated goods was subsequently improved at the expense to achieving $(1-\varepsilon)$-EFX, instead of exact EFX \citep{CGMMM21, ACGMM22, BBK22, CSJS23}.   
    Also, \citet{ARS22} proposed notion EF2X, another relaxation of EFX, and proved its existence when agents have restricted additive valuations.   
    Further, \citet{KSS24} proved the same result for $(\infty,1)$-bounded valuations.

    The most well-studied relaxations of proportionality is the maximin fair share (MMS) notion. The maximin fair share for an agent is the maximum value that she could achieve if she were to partition the goods into $n$ bundles, but then be allocated the one she values the least among them. An allocation is MMS if every agent receives a bundle that they value at least as much as their maximin fair share. 
    The existence of $\frac{2}{3}$-MMS allocations was shown using several different approaches \citep{AMNS15, BK20, KPW18}, and subsequently improved to $\frac{3}{4}$-MMS \citep{GHSSY18},  $\frac{3}{4}+\frac{1}{12n}$-MMS \citep{GT20}, $\frac{3}{4}+\min(\frac{1}{36},\frac{3}{16n-4})$-MMS \citep{AGST23},
    and finally $\frac{3}{4}+\frac{3}{3836}$-MMS \citep{AG24}.
    In very recent independent work,  \citet{HAR24} showed that there exist allocation satisfying both $\frac{2}{3}$-MMS and EF1 and they also proved the existence of epistemic EFX allocations for monotone valuations \citet{AR24}.
    
    Other notions related to MMS include the pairwise maximin share guarantee (PMMS) \citep{CKMPSW19}, in which the maximin share guarantee is required only for pairs of agents rather than the grand bundle; its existence remains open. On the other hand, the GMMS notion \citep{BBMN18} generalizes MMS, requiring that the maximin share guarantee is achieved not just with respect to the grand bundle, but also for all subsets of agents. Other relaxations of proportionality include PROP1~\cite{CF017}, PROPx~\cite{AMS20}, and PROPm \cite{BGGS21}.

    A broader overview of discrete fair division can be found in a survey by \citet{AABRLMVW23}.

\section{Preliminaries}\label{sec:prelims}
    Let $[\ell]$ for some $\ell\in\mathbb{N}$ denote the set $\{1,2,\ldots,\ell\}$ and let $[0]=\emptyset$. 
    We are given a tuple $\langle N,M,V\rangle$, 
    where $N=[n]$ is a set of agents,
    $M$ is a set of indivisible goods, 
    and $V= (v_1, v_2, \ldots, v_n)$ determines the valuation function 
    $v_i: 2^M \to \mathbb{R}_{\geq 0}$ of each agent $i\in N$ for each bundle of goods $S\subseteq M$.

    For simplicity, we denote the union of a bundle $S$ with a good $g\in M$ using $S \cup g$ instead of $S \cup \{g\}$; we also use $S\setminus g$ to denote $S \setminus \{g\}$. We say that $\X= (X_1,\ldots,X_k)$, for some $k\in[n]$, is a partition of the goods if $X_i\subseteq M$ for all $i$, 
    $X_i \cap X_j = \emptyset$ for all distinct $i,j$, and $\bigcup_{j=1}^k X_j = M$.
    Throughout the paper, we use bold capital letters like $\X,\xp,\xz,$ and $\Z$ to denote partitions. When a partition $\X$ has $k$ bundles, we also denote it as $\X^k$ if we wish to make its size explicit. We use $\Y$ and $\yp$ to denote a set of bundles that may not necessarily be a partition (i.e., these bundles may be overlapping and their union may not be $M$). For a partition $\X$, we denote its $\ell$'th bundle with $X_\ell$. Apart from that, we also use $\RE,S,T$ to denote bundles without specifying a partition that they are part of. 
    
    To capture the agents' preferences over bundles, we also use $S \lowerval{i} T$ 
    when $v_i(S) < v_i(T)$ and $S\lowereqval{i} T$ when $v_i(S)\leq v_i(T)$. For a set of bundles $\Y$, we let 
    $\argmax_i(\mathbf{Y})\in \arg\max_{Y \in \Y} v_i(Y)$ denote the most valuable bundle in $\Y$ from agent $i$'s perspective. If there are multiple such bundles, this tie-breaks in favor of bundles with more goods, and it arbitrarily tie-breaks among bundles of equal value and size. Similarly, we use $\textsc{NextBest}_i(\mathbf{Y})$ to denote the second most valuable bundle in $\Y$ from agent $i$'s perspective, and  $\argmin_i(\Y)$ 
    to denote the least valuable bundle in $\Y$ from agent $i$'s perspective, once again using the same tie-breaking rule.

    \paragraph{Types of Valuations.} 
    We consider valuation functions that are monotone, i.e., for any $S\subseteq T\subseteq M$, we have $v_i(S) \leq v_i(T)$. One of the most well-studied classes of valuation functions is \emph{additive}: a valuation function $v$ is additive if the value for any bundle $S$ equals the sum of the values of its goods, i.e., $v(S)=\sum_{g\in S} v(g)$. 
    In this paper, we introduce a new class of valuation functions, called
    \emph{restricted MMS-feasible}, that are even more general than additive; they also extend unit demand, budget-additive, and multiplicative valuations 
    (see \cref{valuation section} for a proof and related discussion).
    \begin{definition}[restricted MMS feasible]
        A valuation function $v : 2^M \to \mathbb{R}_{\geq 0}$ is restricted-MMS-feasible 
        if for any bundle $S \subseteq M$, 
        natural number $k$, 
        and any two partitions
        $\X^k$ and
        $\Z^k$ of $S$, we have:
        \begin{align*}
            \max(v(X_1), \ldots, v(X_k)) \geq \min(v(Z_1),\ldots, v(Z_k)).            
        \end{align*}
    \end{definition}

\paragraph{Envy, EFX-envy, and EFL-envy.}    
    We say agent $i$ envies bundle $T$ \with\ bundle $S$ if $S\lowerval{i} T$ 
    and that agent $i$ \efxenvies\ bundle $T$ \with\ $S$ if $S\lowerval{i} T\setminus g$ for some good $g \in T$. 
    We say agent $i$ \eflenvies\ bundle $T$ \with\ $S$ if $|T|>1$ and for every $g \in T$, either $S\lowerval{i} T \setminus g$ or $S\lowerval{i} g$.

\paragraph{EFX, EFL, EEFX, and MXS} 
    We say that a bundle $X_\ell$ is \efxf\ for agent $i$ in some partition $\X$ 
    if agent $i$ does not \efxenvy\ any other bundle in $\X$ \with\ bundle $X_\ell$. 
    We say that bundle $X_\ell$ in partition $\X^k$ is k-\eefxf\ for agent $i$ 
    if there exists a partition $\Y^k$ (with the same number of bundles as $\X^k$) 
    such that $Y_\ell$ is \efxf\ for agent $i$ in $\Y$, and $X_\ell=Y_\ell$.
    We say that bundle $X_\ell$ in partition $\X^k$
    is k-\mxsf\ for agent $i$, if there exists a partition $\Y^k$ (with the same number of bundles as $\X^k$)
    such that $Y_\ell$ is \efxf\ for agent $i$ in $\Y^k$, and $X_\ell \greatereqval{i} Y_\ell$.
    We say that a bundle  $X_\ell$ is \eflf\ for agent $i$ in partition $\X$ 
    if agent $i$ does not \eflenvy\ any bundle in $\X$ \with\ bundle $X_\ell$.
    We say that a bundle $X_\ell$ is k-\eflmf\ in $\X^k$ if it is both k-\mxsf\ and \eflf\ in $\X^k$.
    When the number of bundles is fixed and obvious from the context, we may drop $k$ 
    and just say \eefxf, \mxsf\ and \eflmf.

\paragraph{\AssociationFunction s and Allocations.}
    We say a function $f: [n] \to [n]\cup \{\empt\}$ is an \emph{association function} if for any two distinct $\ell,z\in [n]$ we have either $\f{\ell}\ne \f{z}$ or $\f{\ell}= \f{z}=\empt$. We denote by $D(f)=\{\ell \in [n]: f(\ell)\neq \empt\}$ the support of $f$ and by $\range (f)=\{i\in [n]: \exists \ell \in [n] \text{ s.t.\ } f(\ell)=i\}$ the image of $f$. An \emph{allocation} $(\X,f)$ is the combination of a partition $\X$ with an association function $f$ that associates each bundle $X_\ell$ of $\X$ either with an agent (when $f(\ell) \in [n]$) or with no agent (when $f(\ell)=0$). If $f(\ell) \in [n]$, we say that bundle $X_\ell$ and agent $f(\ell)$ are associated with each other. On the other hand, if $\ell \notin D(f)$, we say bundle $X_\ell$ is ``free,'' and if $i\notin \range (f)$ we say that agent $i$ is ``free.'' Given an allocation $(\X,f)$, the value of each agent $i\in \range(f)$ is their value for the bundle of $\X$ associated to them according to $f$. We say an association function $f$ is \mef\ for $\X$ if for every $\ell \in D(f)$ the bundle $X_\ell$ is \eflmf\ for agent $\f{\ell}$ in $\X$. Also, we say $f$ is \emph{full} for $\X$ if every bundle of $\X$ is associated with an agent, i.e., $D(f)=[k]$, where $k=|\X|$.

\paragraph{Generalized Envy Graph.}
	Given a partition $\X$ and an association function $f$, the Generalized Envy Graph 
    of $(\X, f)$ is a directed graph $G(\X,f)=(V,E)$ whose set of vertices $V$ corresponds to the set of bundles in $\X$, 
    and a directed edge $(X_\ell, X_z)$ exists in $E$ if and only if 
	$\f{X_\ell} \ne \empt$ and $X_\ell \lowerval{\f{X_\ell}} X_z$, i.e., bundle $X_\ell$ is associated with some agent who envies $X_z$ \with\ $X_\ell$. 
    We say a bundle $X_\ell$ is a source in $G(\X,f)$ if the in-degree of $X_\ell$'s vertex in the graph is zero, so 
    no agent in $\range (f)$ envies a source bundle \with\ their associated bundle. To denote a path $(X_{q_s}\to \ldots X_{q_1} \to X_{q_0})$ in $G(\X,f)$, we also use just the sequence of indices $(q_s,\ldots,q_1,q_0)$. Note that throughout the paper we focus on the last node of such paths, which is why we enumerate their vertices backwards.

\paragraph{Envy Cycle Elimination.}
    The \textsc{EliminateCycles} procedure takes as input an allocation $(\X,f)$ and returns an association function $f'$ such that the \gegraph\ $G(\X,f')$ does not contain any cycles. To achieve this, while there still exist a cycle in the \gegraph, this procedure essentially performs a ``rotation'' along this cycle, i.e., for each directed envy edge $(X_\ell, X_z)$ in the cycle it updates the association function so that $X_z$ is associated with the agent who was previously associated with bundle $X_\ell$, i.e., the agent who envied  $X_z$ relative to $X_\ell$. This is a type of procedure commonly used to remove cycles of envy, starting from \cite{LMMS04}, but since we use it in a more general family of graphs, we prove its termination guarantee and other properties in \cref{envy cycle elimination ends}.
    

\begin{algorithm}[H]
 \SetAlgoRefName{EliminateCycles}
    \SetAlgoNoEnd
    \SetAlgoNoLine
    \DontPrintSemicolon
    \NoCaptionOfAlgo 
    \KwIn{partition $\X$, association function $f$}    
		\While{ $\exists$ cycle $(q_s, q_{s-1},\ldots q_1, q_s)$ in $G(\X,f)$}{
				$i \gets f(q_1)$
                
	 		\For{ $\ell=1$ to $s-1$}{
				$f(q_\ell) \gets f(q_{\ell+1}) $
				}
				$f(q_s) \gets i$
	}
    \Return{$f$}
\caption{\textsc{EliminateCycles}}
\label{envy elimination alg}
\end{algorithm}


\paragraph{Chains and subchains.} Given an \gegraph\ $G(\X,f)$, we use the term \emph{chain} to refer to a simple\footnote{That is, a path that includes each vertex of the graph at most once.} path $(q_s,q_{s-1},\ldots,q_1,q_0)$ in $G(\X,f)$ whose first node, $X_{q_s}$, is a source, i.e., has no incoming edges. More broadly, we refer to any simple path, without the requirement that its first vertex is a source, as a \emph{subchain}, and for every chain $c=(q_s,\ldots,q_0)$ and every $0\leq \ell \leq s$, we say $c'=(q_\ell,\ldots,q_0)$ is a subchain of $c$, or that $c$ contains $c'$.  
We also use $X_{q_s} \rightsquigarrow X_{q_0}$ to denote a subchain $(q_s,\ldots,q_0)$, and we refer to $X_{q_s}$ as its ``first'' bundle. The number of edges of a subchain, $s$, is its length, and it may be zero, i.e.,\ $(q_0)$ is a subchain. 

\begin{definition}[Set of chains]
\label{def set of chains}
    Given a partition $\X$ and an association function $f$,
    we use $C_k(\X,f)$ to denote the set of chains of $G(\X, f)$ that terminate at $X_k$, where $k=|\X|$, i.e.,:
    \begin{center}
        $C_k(\X,f)= \Big\{ 
        c=(q,\ldots,k) :~  c \text{ is a chain to } X_{k} \text{ in } G(\X,f) \Big\}$
    \end{center}
    We also define set of subchains of $G(\X,f)$ that terminate at $X_k$ as:
    \begin{center}
        $SC_k(\X,f)= \Big\{c=(q,\ldots,k):~ c \text{ is a subchain to } X_{k} \text{ in } G(\X,f) \Big\}$
    \end{center}
\end{definition}

\paragraph{Shifting a subchain.} 
    The \textsc{ShiftSubChain} procedure takes as input an association function $f$ and a subchain $c=(q_s,\ldots,q_0)$ and returns an $f'$ such that for each directed edge $(X_\ell,X_z)$ in $c$, the agent $f(\ell)$ is associated with bundle $X_z$ instead, leaving free bundle $X_{q_s}$ and agent $f(q_0)$ (if any). 
    
\begin{algorithm}[H]
 \SetAlgoRefName{ShiftSubChain}
\SetAlgoNoEnd
\SetAlgoNoLine
\DontPrintSemicolon
\NoCaptionOfAlgo 
\KwIn{$f, c=(q_s,\ldots,q_0) $}     
    \For{$\ell=0$ to $\ell=s-1$}{
        $f(q_\ell) \gets f(q_{\ell+1})$}
    $f(q_s)\gets 0$
    
    \Return{$f$}  
\caption{\textsc{ShiftSubChain}}
\label{shift subchain alg}
\end{algorithm}

\section{Algorithm Constructing MXS+EFL Allocations}

    We now provide an algorithm that takes as input $M$ and any monotone and restricted-MMS-feasible agent valuations, and constructs a partition $\X^n$ and a full association function $f$ for $\X^n$ which simultaneously satisfies MXS and EFL.
    
    \begin{theorem}\label{maintheorem}
        If all the valuations are monotone and restricted-MMS-feasible, \ref{MXS+EFL} always terminates and returns an allocation that is simultaneously MXS and EFL for every agent.
    \end{theorem}

    The MXS+EFL algorithm gradually constructs $\X^n$ and $f$ over a sequence of $n-1$ iterations of its for-loop (from $k=2$ to $k=n$). At the beginning of each iteration, the algorithm has constructed a partition $\X^{k-1}$ and a \fmefaf\ $f$ for $\X^{k-1}$ (starting from $\X^1=(M)$ and $f(1)=1$ for $k=2)$. It then augments $\X^{k-1}$ to create $\X^k$, using an empty bundle for $X^k_{k}$; this breaks the ``balance'' of fairness it had achieved in the smaller partition. To restore balance in the bigger partition, it then calls a process called \textsc{ReBalance}, which does all of the heavy lifting, by carefully reallocating the goods across the $k$ bundles and updating the association function to reach a new partition $\X^k$ along with a \fmefaf\ $f$ for $\X^k$. 
    Therefore, the proof of \cref{maintheorem} reduces to proving that \textsc{ReBalance} restores fairness for every $k$ (\cref{thm:rebalance_main}).
    
\vspace{10pt}

\begin{algorithm}[H]
 \SetAlgoRefName{Algorithm MXS+EFL}
\SetAlgoNoEnd
\SetAlgoNoLine
\DontPrintSemicolon
    \NoCaptionOfAlgo 
\KwIn{$ V=(v_1, v_2,\ldots, v_n)$ and $M$}    
   $\X^1 \gets (M)$, ~~$f(1) \gets 1$, ~~ and $f(\ell) \gets 0$ for all $\ell \in \{2,3,\dots, n\}$ \label{base induction alg} 
    
    \For{ $k=2$ to $n$}{   \label{induction alg}
        $f\gets$ \textsc{EliminateCycles}$(\X^{k-1},f)$   \label{envy elimination before proc 1}
        
        
        $\X^{k} \gets (X^{k-1}_1,\ldots, X^{k-1}_{k-1}, \emptyset)$  \\ \label{add empty bundle} 
        
        
        $(\X^k,f) \gets$\textsc{ReBalance}($\X^{k},f$)   \label{rebalance input}
  }
  \Return $(\X^n, f)$
\caption{MXS+EFL}
\label{MXS+EFL}
\end{algorithm}

\vspace{10pt}

We first provide a very high-level and ``loose'' description of the \textsc{ReBalance} algorithm to provide some intuition, and we defer the details until after we also introduce some additional subroutines and notation. The \textsc{ReBalance} algorithm takes place in two phases, each of which corresponds to a while-loop in its pseudocode (\cref{first loop alg} to \cref{break 1} for Phase 1, and \cref{second loop alg} to \cref{change X proc 2} for Phase 2). 

\textbf{Phase 1:} In each iteration of the Phase 1 while-loop, the algorithm first checks whether there exists a \fmefaf\ $f$ for the current partition $\X$ and returns $(\X, f)$ if it does. If not, it chooses some free agent $i$, it identifies a bundle $\Qa\in \X$ that this agent ``likes,'' and it also identifies the agent $j$ currently associated with $\Qa$. Then, the algorithm identifies an ``under-demanded'' bundle $\Qb$ and if one of the two agents $u\in \{i,j\}$ would prefer $\Qa$ over $\Qb$ even if one of $\Qa$'s goods were to be moved to $\Qb$, then this good is moved (updating $\X$), agent $u$ is associated with $\Qa$ (updating $f$), and we proceed to the next iteration. If, on the other hand, both of the agents would prefer $\Qb$ after moving that good, the algorithm frees up both agents $i,j$ and both bundles $\Qa,\Qb$, and proceeds to Phase 2, aiming to resolve this ``tension'' between the two agents. 

\textbf{Phase 2:} Each iteration of the Phase 2 while-loop starts with a partition $\X$ and it considers two alternative partitions $\X'$ and $\X''$ that can be reached from $\X$ by removing some good from the ``highly-demanded'' bundle $\Qa$ and moving it to some ``under-demanded'' bundle. If both of these alternatives fail to have a \fmefaf, the algorithm updates $\X$ to be equal to $\X'$ and proceeds to the next iteration. Before doing so, though, it evaluates the extent to which the two agents' preferences may have changed in this new partition and it appropriately updates $\Ip$ and $\Iq$, to point to the ``highly-demanded'' and ``under-demanded'' bundle, respectively; if these bundles changed, it also uses the \textsc{ShiftChain} subroutine to update $f$ and free up $\Qa$ and $\Qb$. 

Throughout both Phase 1 and Phase 2, \textsc{ReBalance} only considers partitions of the same size and it only terminates if one of these partitions that it reaches has a \fmefaf. Therefore, \textsc{ReBalance} is guaranteed to regain ``balance'' if it terminates, and the main technical obstacle is to prove that it does, indeed, always terminate. We dedicate both \cref{sec:phase1} and \cref{sec:phase2} to proving this non-trivial fact using carefully chosen invariants and potential functions which verify that the \textsc{ReBalance} algorithm keeps ``progressing'' and will not get stuck in an infinite loop.

To understand each step of the \textsc{ReBalance} pseudocode, we now also introduce some additional subroutines and notation (beyond those in \cref{sec:prelims}) that the algorithm uses.

\paragraph{Subroutines Used in \textsc{ReBalance}} The \textsc{HasFairAssociation} subroutine takes as input a partition $\X$ and it returns ``true'' if there exists a \fmefaf\ for $\X$ and ``false'' otherwise. \textsc{FairAssociation} takes as input a partition $\X$ for which \textsc{HasFairAssociation} is ``true,'' and it returns a \fmefaf\ for $\X$. Note that since we just need to show the existence of MXS+EFL allocations, we are not worried about the running time of these subroutines; as an extreme example, even if we were to exhaustively search through all possible associations, this would terminate in finite time. Finally, the \textsc{FindMins} subroutine is given a bundle $\PA$ and two agents $i,j$, and it returns two goods $\xii,\xj\in \PA$ that provide the smallest marginal value to $i$ and $j$, respectively. If there are multiple such items, it returns an arbitrary one for each agent, unless the agents share an item of minimum marginal value, in which case this item is returned for both.

\vspace{5pt}

    \begin{algorithm}[H]
 \SetAlgoRefName{FindMins}
    \SetAlgoNoEnd
    \SetAlgoNoLine
    \DontPrintSemicolon
    \NoCaptionOfAlgo 
    \KwIn{$(\PA,i,j)$}
    Let $\textsc{Min}_i\gets \arg\max_{g \in \PA} v_i(\PA \setminus g)$ ~~and~~ $\textsc{Min}_j\gets \arg\max_{g \in \PA} v_j(\PA \setminus g)$\;
    \lIf{ $\textsc{Min}_i \cap \textsc{Min}_j \neq \emptyset$}{
        Let $\xii=\xj \in \textsc{Min}_i \cap \textsc{Min}_j$
    }\lElse{
        Let $\xii \in \textsc{Min}_i$, ~~and~~ $\xj \in \textsc{Min}_j$
    } 		
    \Return{$\xii,\xj$}
\caption{\textsc{FindMins}}
\label{Choose Leasts}
\end{algorithm}

\paragraph{Some Additional Notation.}
    Given a bundle $\PA$ and some agent $i$, 
    we let $\xii \in  \arg\max_{g \in \PA} v_i(\PA \setminus g)$ denote any good in $\PA$ with the smallest marginal value for $i$.
    We say bundle $X_z$ is \emph{\bstf} for agent $i$ in some partition $\X$ 
    if for every other bundle $X_\ell$ in $\X$,
    we have  either $X_\ell = \emptyset$ or
    $v_i(X_z \setminus x_z^i) \geq  v_i(X_\ell \setminus x_\ell^i)$. 
    We use $\text{\bstfset}_i(\X)$, $\text{\efxfset}_i(\X)$, $\text{\eflfset}_i(\X)$,
     $\text{\mxsfset}_i(\X)$, and $\text{\eflmfset}_i(\X)$  
    to denote the set of \bstf, \efxf, \eflf, \mxsf, and \eflmf\ bundles
    for agent $i$ in partition $\X$, respectively. 
    E.g., we write \isefxffor{X_\ell}{i}{\X}, when
    bundle $X_\ell$ is \efxf\ for agent $i$ in $\X$. 

\begin{theorem}\label{thm:rebalance_main}
    Whenever the \textsc{ReBalance} algorithm is called with input $(\X^k,f)$ such that $f$ is \mef\ for $\X^k$, $D(f)=[k-1]$, and $G(\X^k,f)$ is acyclic, it terminates within a finite number of steps and it returns a partition $\X^k$ and an $f$ such that $f$ is \mef\ for $\X^k$ and $D(f)=[k]$.
\end{theorem}
\begin{proof}
    The fact that the returned $f$ is \mef\ for $\X^k$ and $D(f)=[k]$ is directly implied from the fact that \textsc{ReBalance} only terminates if it identifies a bundle $\X$ that satisfies the \textsc{HasFairAssociation} condition and, when it does, it returns that bundle along with a \fmefaf\ for it. The main difficulty therefore lies in proving that \textsc{ReBalance} always terminates within a finite number of steps. We dedicate \cref{sec:phase1} to proving that the Phase 1 terminates  (\cref{thm:phase1terminates}) and \cref{sec:phase2} to proving that Phase 2 also always terminates  (\cref{main lemma of proc 2}).   
\end{proof}

\begin{lemma}\label{lem:rebalance_input}
    Whenever \textsc{ReBalance} is called by the MXS+EFL algorithm, its input $(\X^k,f)$ is such that such that $f$ is \mef\ for $\X^k$, $D(f)=[k-1]$, and $G(\X^k,f)$ is acyclic.
\end{lemma}
\begin{proof}
    The fact that $G(\X^k,f)$ is acyclic is implied by the fact that before calling $\textsc{ReBalance}(\X^k,f)$ the MXS+EFL algorithm calls $\textsc{EliminateCycles}(\X^{k-1},f)$, combined with the fact that $G(\X^k,f)$ and $G(\X^{k-1},f)$ have the same set of edges (the new bundle, $X^k_k$, is empty and has no incoming or outgoing edges). Therefore, the acyclicity of $G(\X^{k-1},f)$ implies the acyclicity of $G(\X^k,f)$.
    
    We now prove that $f$ satisfies the other two properties using induction on $k$. For $k=2$, we have $\X^2=(M,\emptyset)$ and $f(1)=1$, i.e., agent $1$ is allocated all the goods. Therefore, $D(f)=[k-1]=\{1\}$ and $f$ is both MXS and EFL for $\X^2$. Now, consider any $k>2$ and assume the statement of the lemma holds for the call \textsc{ReBalance}$(\X^{k-1},f)$ of the previous iteration. \cref{thm:rebalance_main} implies that the output $(\X^{k-1},f)$ of this call satisfies $D(f)=[k-1]$ and that $f$ is MXS+EFL for $\X^{k-1}$. To conclude the proof, we observe that this implies that in the subsequent call to \textsc{ReBalance}$(\X^k,f)$ we still have $D(f)=[k-1]$, and we also have that this $f$ is MXS+EFL for $\X^k$ since the set of bundles remains the same and we just added an empty bundle. This is clearly unaffected by any cycle eliminations in \cref{envy elimination before proc 1}.
\end{proof}

\begin{algorithm} 
 \SetAlgoRefName{Rebalance}
 \SetAlgoNoEnd
 \SetAlgoNoLine
 \DontPrintSemicolon
 \NoCaptionOfAlgo 
 \KwIn{$\big(\X,f\big)$}  
\While{\tikzmark{top}}{    \label{first loop alg}
	\If{$\textsc{HasFairAssociation}(\X)$}{
        \Return $(\X,\textsc{FairAssociation}(\X))$     \\     \label{i chain satisfied return proc 1}
        }
    Let $i\in N\setminus \range (f)$,~~ \isbstffor{\Qa}{i}{\X}, ~~and~ $j \gets f(p)$ 
\qquad\tikzmark{right}\\  \label{getting best for i} 
            
    \If{ \Isnotbstf{\Qa}{j}{\X}  }{      \label{X_z not bstf for j:proc 1}
  		$f(p) \gets i$ ~and~ \textbf{repeat} \\   \label{repeat 1.1}
        }
    $\xii,\xj \gets$ \textsc{FindMins}$(\Qa,i,j)$ \label{choice of xzu} \\

    Let $k\gets |\X|$ be the number of bundles in $\X$ \\
    Let $c$ be any chain $\QB \rightsquigarrow X_k$ in $G(\X,f)$ \label{choose chain:proc 1 alg} \\ 

	\If{for some $u \in \{i,j\}$ we have $\Qb \cup \xu \lowereqval{u} \Qa \setminus \xu$}{   \label{case 1 proc 1}
        $\Qa\gets \Qa\setminus \xu$, \hspace{2mm} $\Qb \gets \Qb \cup \xu$,
        \hspace{2mm} $f(p)\gets u$  \label{change in case 1}
        
        $f\gets$\textsc{EliminateCycles}$(\X,f)$  ~and~ \textbf{repeat} \label{envy elimination proc 1}
      
	}
    \textbf{break} \tikzmark{bottom}  \label{break 1}
    }

$f\gets$\textsc{ShiftSubChain}$(f,c)$ ~and~ $f(p)\gets 0$      \label{change f before phase 2}

$\RE = \PA \setminus \xj$  \label{E def first}


\While{\tikzmark{top2}}{  \label{second loop alg}
    $f\gets$\textsc{EliminateCycles}$(\X,f)$ \label{envy elimination line}
    
    Let $c$ be any chain $\Qc \rightsquigarrow \Qb$ in $G(\X,f)$   \label{choose c proc 2}

    $\xii,\xj \gets$ \textsc{FindMins}$(\Qa,i,j)$ \label{choice of xzu proc2} \\
    Let $\xp\gets \X$,~~ $\Qa' \gets \Qa\setminus \xj$,~~ and~ $\Qc'\gets \Qc\cup \xj$\\

    \If{\textsc{HasFairAssociation}$(\xp$)}{
        \Return $(\xp, \textsc{FairAssociation}(\xp))$     \\     \label{i chain satisfied return proc 2 X'}
        }
    Let $\xz\gets \X$,~~ $\Qa''\gets \Qa\setminus \xii$,~~ and~ $\Qc''\gets \Qc\cup \xii$\\    \label{construct part X''}
    \If{$\textsc{HasFairAssociation}(\xz)$}{
        \Return $(\xz, \textsc{FairAssociation}(\xz))$     \\     \label{i chain satisfied return proc 2 X''}
        }
            $\RE \gets \textsc{NextBest}_j (\RE,\PA',X_r')$ \label{update E}
            
    \If{$\Qc' \greaterval{j} \Qa'$}{  \label{j gets max}
            $f\gets$\textsc{ShiftSubChain}$(f,c)$ and let $\Iq \gets \Ip$ and $\Ip\gets \Ir$     \label{PA'<j SC'}
        }
    \ElseIf{$\Qc' \greatereqval{i} \Qb$}{    \label{i gets second max}
            $f\gets$\textsc{ShiftSubChain}$(f,c)$ and let $\Iq \gets \Ir$    \label{q gets i second max}
    }
    $\X \gets \xp$  \tikzmark{bottom2}\label{change X proc 2}}
\caption{\textsc{ReBalance} algorithm}
\label{Rebalance}
\AddNote{top}{bottom}{right}{Phase 1}
\AddNote{top2}{bottom2}{right}{Phase 2}
\end{algorithm}

\section{Some Useful Lemmas and Observations}
\label{lemmas main body}
    We now provide some lemmas and observations that we use throughout the paper. We defer their proofs to \cref{lemmas appendix}.

\begin{restatable}{observation}{eliminatecyclesend}
\label{envy cycle elimination ends}
    The \textsc{CycleElimination} procedure always terminates and if its input $(\X, f)$ is such that $f$ is an \mef\ association function for $\X$, then the $f$ it returns remains \mef\ for $\X$ and  also $D(f)$ and $\range(f)$ remain the same.
\end{restatable}

\begin{restatable}{lemma}{lemone}
\label{envy subchain}
    If the \textsc{ShiftSubChain} procedure takes as input an association function $f$ 
    that is \mef\ for some partition $\X$,  the $f'$ that it returns is also \mef\ for $\X$. Also, if the subchain in the input is $X_{q} \rightsquigarrow X_p$, then $D(f')= (D(f)\cup \{p\})\setminus \{q\}$ and $\range(f')\subseteq \range(f)$.   Furthermore,  for every agent in $\range(f')$, their value in the assignment $(\X, f')$ is weakly greater than their value in $(\X, f)$. Finally, if some bundle $\SC$ is a source in $G(\X,f)$, it will also be a source in $G(\X,f')$.
\end{restatable}

\begin{restatable}{observation}{obsone}
\label{hierarchy}
    Given any partition $\X$ and any agent $i$, all the following statements are true:
    \begin{itemize}
        \item  If \isbstffor{X_z}{i}{\X}, then \isefxffor{X_\ell}{i}{\X} if and only if $X_\ell \greatereqval{i} X_z \setminus x_z^i$.
        \item  If \isbstffor{X_z}{i}{\X}, then \isefxffor{X_z}{i}{\X}.
        \item  If \isefxffor{X_z}{i}{\X}, then \iseflmffor{X_z}{i}{\X}.
        \item  If \iseefxffor{X_z}{i}{\X}, then \ismxsffor{X_z}{i}{\X}. 
    \end{itemize}
\end{restatable}

\begin{restatable}{observation}{obstwo}
\label{EFX-Best prop to remove good}
    Given any partition $\X$, any agent $i$, and \isbstffor{\PA}{i}{\X}, 
    if $\X'$ is a partition with $X'_p \gets X_p\setminus \xii$, $X'_r \gets X_r \cup \xii$ and $\X'_\ell=X_\ell$ for all $\ell \notin \{p,r \}$, then
    agent $i$ does not \efxenvy\ any bundle in $\xp \setminus \{X_r \cup \xii\}$ \with\ bundle $\PA \setminus \xii$.
    Additionally, \isefxffor{\argmax_i(\PA \setminus \xii, X_r \cup \xii)}{i}{\X'}.
\end{restatable}

\begin{restatable}{lemma}{lemoneEFLtwo}
    \label{EFL two}
    Suppose $\X$ is a partition, \iseflmffor{X_\ell}{i}{\X}, 
    and let $\PA$ and $\SC$ be any two 
    bundles in $\X$
    such that $X_\ell \greatereqval{i} \SC$ and $|\PA| \geq 2$.
    Then, we consider two cases:
    \begin{itemize}
        \item \textbf{Case 1:} If $p \ne \ell$ and we remove some good $g$ from $\PA$ and add it to $\SC$, leading to partition $\X'$, 
    							then \iseflmffor{X'_\ell}{i}{\xp}.
        \item \textbf{Case 2:} If $p = \ell$,  
          $\PA \setminus \xii \greatereqval{i} \SC$, and \isbstffor{\PA}{i}{\X},
    								 then if we remove good $\xii$ from $\PA$ and add to $\SC$, leading to partition $\xp$,
    								 then \iseflffor{X'_p}{i}{\xp}. 
    \end{itemize}
\end{restatable}

\begin{restatable}{observation}{chainexists}
\label{chain existence}
    If there is no cycle in an \gegraph\ $G(\X,f)$, 
    then for every bundle $X_\ell \in \X$
    there exists a chain in $C_\ell(\X,f)$.
\end{restatable}

\section{Termination of Phase 1 of \textsc{ReBalance}}
\label{sec:phase1}

\begin{theorem}\label{thm:phase1terminates}
\label{1 ends}
    Phase 1 of the \textsc{ReBalance} algorithm terminates after a finite number of iterations.
\end{theorem}

In this section we will be comparing subchains of different allocations, so we introduce the following definitions.
\begin{definition}[Subchain equality]
\label{subchain equality}
    We say a subchain $c=(q_{s},\ldots, q_1, k)\in SC_k(\X,f)$ is \emph{equal} to a subchain $c'=(q'_{s},\ldots, q'_1, k) \in SC_k(\X', f')$ if $X_{k}=X'_{k}$ and for every $\ell\in [s]$ we have $X_{q_\ell}=X'_{q'_\ell}$ and $f(q_\ell)=f'(q'_\ell)$. When this holds, we denote it as $c=c'$.
\end{definition}

    Using this definition, we can now also define a subset relationship across sets of subchains induced by different allocations.


\begin{definition}[Comparing sets of subchains]
\label{subset chain equality}
    We say a set of subchains $SC_k(\X,f)$ is a \emph{subset} of another set of subchains $SC_k(\X',f')$ if for every $c\in SC_k(\X,f)$ there is a $c'\in SC_k(\X', f')$ such that $c=c'$. If it is also true that for every $c'\in SC_k(\X',f')$ there is a $c\in SC_k(\X, f)$ such that $c=c'$, then the two sets are equal. If not, then $SC_k(\X,f)$ is a \emph{strict} subset of $SC_k(\X',f')$. We denote these as $SC_k(\X,f) \subseteq SC_k(\X',f')$, $SC_k(\X,f) = SC_k(\X',f')$, and $SC_k(\X,f) \subset SC_k(\X',f')$, respectively. This can be extended to define $C_k(\X,f) \subseteq C_k(\X',f')$, $C_k(\X,f) = C_k(\X',f')$, and $C_k(\X,f) \subset C_k(\X',f')$ over sets of chains.
\end{definition}

\subsection{Phase 1 invariant}

To prove \cref{thm:phase1terminates}, we show that the following invariant always holds in Phase 1,
and then use it to prove \cref{thm:phase1terminates} in the next section.

\begin{definition} [\textbf{Phase 1's invariant}]
    We say that allocation $(\X,f)$ satisfies Phase 1's invariant if 
    $f$ is an \mef\ association function for $\X$ such that $D(f)=[k-1]$, where $k=|\X|$. Also, 
    every subchain $c'\in SC_k(\X,f)$ is contained in a chain $c''\in C_k(\X,f)$.
\end{definition}

\begin{restatable}{lemma}{phaseainvhold}
\label{phase 1 inavariant holds}
    At the start of every iteration of phase 1,
    $(\X,f)$ satisfies Phase 1's invariant.
\end{restatable}
\begin{proof}
    By \cref{lem:rebalance_input}
    at the first iteration of phase 1, $(\X,f)$ is such that such that $f$ is \mef\ for $\X^k$, $D(f)=[k-1]$, and $G(\X^k,f)$ is acyclic,
    so by \cref{chain existence}, phase 1's invariant holds at the start of first iteration.
    Phase 1 may repeat itself in \cref{repeat 1.1} or \cref{envy elimination proc 1}.
    So, we use induction and assume that phase 1's invariant holds at the start of one iteration and show that
    it holds at the start of next iteration.
    We will prove this in \cref{repeat bstf} if phase 1 repeats in \cref{repeat 1.1},
    and in \cref{repeat proc 1 case 1 ok} if phase 1 repeats in \cref{envy elimination proc 1}.
\end{proof}

    In this section, by $i,j,p$, we mean those that were defined in \cref{getting best for i}, 
    by $k$ we mean the number of bundles in $\X$,
    $f$ is the association function, and $c$ is the chain $\QB \rightsquigarrow X_k$ in $G(\X,f)$ defined in \cref{choose chain:proc 1 alg}.

\begin{lemma}
\label{p not k proc 1}
    If phase 1's invariant holds at the start of phase 1 while loop, and if \textsc{ReBalance} reaches \cref{getting best for i}, then $p \ne k$.
\end{lemma}
\begin{proof}
       If the \textsc{ReBalance} reaches \cref{getting best for i}, 
    then it has not returned in \cref{i chain satisfied return proc 1}, 
    so \Isnoteflmf{X_k}{i}{\X}, and therefore, \Isnotbstf{X_k}{i}{\X}.
    Since if \iseflmffor{X_k}{i}{\X}, then by setting $f(k)=i$, we get a \fmefaf.
    So, since \isbstffor{\PA}{i}{\X}, we get $p \ne k$.
\end{proof}

\begin{restatable}{lemma}{lemthree}
\label{ending:proc 1}
    If phase 1's invariant holds at the start of phase 1 while loop, 
    first while loop is well-defined, and it always terminates. 
\end{restatable}
\begin{proof}
    We first show that there exists a chain $\QB \rightsquigarrow X_k$ in $G(\X,f)$, 
    so \cref{choose chain:proc 1 alg} executes properly.
    Let $c'$ be the subchain that contains only node $X_k$. Then,
    since $(X,f)$ satisfies Phase 1's invariant,
    there exists a chain $c$ that contains $c'$, so this chain exists. 
    At \cref{getting best for i}, \textsc{ReBalance} chooses (arbitrarily) an agent $i$ not in $\range (f)$, 
    which exists, since there are $|\range (f)|=k$ agents in $\range (f)$, $k<n$, and there are $n$ agents. 
    By \cref{p not k proc 1}, $p\ne k$, so $p \in D(f)$ and $j=f(p)$ is well-defined.
    Also, envy cycle elimination in \cref{envy elimination proc 1} 
    ends by \cref{envy cycle elimination ends}.
\end{proof}

\begin{restatable}{lemma}{lemfive}
\label{not chain satisfied 1}
    If phase 1's invariant holds at the start of phase 1 while loop, and if \textsc{ReBalance} reaches  \cref{getting best for i}, then
    for any subchain $X_\ell \rightsquigarrow X_k$ in $G(\X,f)$,  
    and for both agents $u \in \{i,j\}$ defined in  \cref{getting best for i},
    \Isnotefxf{X_{\ell}}{u}{\X}.
\end{restatable}
\begin{proof}
    Suppose $c'$ is a subchain $X_\ell \rightsquigarrow X_k$ in $G(\X,f)$.
    We show that if for any $u \in \{i,j\}$, \isefxffor{X_{\ell}}{u}{\X}, then there exists a \fmefaf\ for $\X$. 
    Thus, \textsc{ReBalance} should have returned in \cref{i chain satisfied return proc 1}, 
    which is a contradiction since \textsc{ReBalance} has reached  \cref{getting best for i}.

    Since $D(f)= [n] \setminus \{k\}$ and $i\notin \range(f)$, by shifting $f$ along subchain $c'$ (by \cref{envy subchain}),
    we would get a \mefaf\ $f'$ with $D(f')= (D(f)\cup \{k\})\setminus \{\ell\}= [k]\setminus \{\ell\}$ and $i \notin \range(f')$.
    
    If \isefxffor{X_{\ell}}{i}{\X}, 
    then \iseflmffor{X_\ell}{i}{\X}. 
    Hence, by setting $f'(\ell) = i$, we get a \fmefaf\ for $\X$, which is a contradiction.

    Now we prove the argument for agent $j$.
    So, suppose there exists subchain $c': X_\ell \rightsquigarrow X_k$ such that \isefxffor{X_\ell}{j}{\X}.
    Since \isbstffor{\PA}{i}{\X} (by \cref{getting best for i}), we get \isefxffor{\PA}{i}{\X}.
    If subchain $c'$ contains vertex $\PA$, then there exists a subchain $\PA \rightsquigarrow X_k$ in $G(\X,f)$ and \isefxffor{\PA}{i}{\X},
    so we get a contradiction by our previous statement.
    
    If subchain $c'$ does not contain vertex $\PA$, then $X_p$ will still be associated with agent $j$ in $f'$.
    So, since \iseflmffor{X_p}{i}{\X} and \iseflmffor{X_{\ell}}{j}{\X},
    by setting $f'(p)= i, f'(\ell) = j$, we get a \fmefaf\ for $\X$.
\end{proof}



\begin{lemma}
\label{bstf inc}
    If an iteration of phase 1 starts with $(\X,f)$ satisfying phase 1's invariant, and
    repeats in \cref{repeat 1.1} with updated input $(\X',f')$, 
    then $SC_k(\X',f')=SC_k(\X,f)$ and $C_k(\X',f')=C_k(\X,f)$.
\end{lemma}

\begin{proof}
    We have that $\X=\xp$. Also, $f$ and $f'$ only differ in the agent associated to bundle $X_p$ such that $f(p)=j$ and $f'(p)=i$.
    We need to prove that every subchain  in $SC_k(\X,f')$ is a subchain in $SC_k(\X,f)$ 
    and also the reverse. We need to prove the same thing for chains.
    
    Suppose $c' = (q_{s},\ldots,q_1,k)$ is an arbitrary subchain in  $SC_k(\X,f')$.
    We first prove that $c'$ does not contain index $p$.
    Suppose on the contrary, then there is a subchain $c''=(p,h_z,\ldots,h_1,k)$ in $SC_k(\X,f')$.
    Because $f'(p)=i$, by definition of subchain, we get $X_p \lowerval{i} X_{h_z}$. 
    Hence, by \isbstffor{X_p}{i}{\X}, we get \isefxffor{X_{h_z}}{i}{\X}.
    Since $\X'=\X$ and $f'$ differs with $f$ only in the agent associated with bundle $X_p$,
    we get that $c'''=(h_z,\ldots,h_1,k)$ is a subchain in $G(\X,f)$, which is a contradiction by \cref{not chain satisfied 1}
    since \textsc{ReBalance} has reached  \cref{getting best for i}.
    
    Hence, $c'$ does not contain $p$,
    $\X'=\X$, and for all bundles in $\X$ except $X_p$, associated agents are the same, so we get that
    $c'$ is a subchain in  $SC_k(\X,f')$, too.
    
    Also, we have $X_p \greatereqval{j} X_{q_{s}}$, 
    since otherwise $c^* = (p,q_{s}, \ldots,q_1, k)$ would be a subchain in  $SC_k(\X,f)$, 
    which  is a contradiction by \cref{not chain satisfied 1} since \isefxffor{X_p}{i}{\X} and \textsc{ReBalance} has reached  \cref{getting best for i}.
    Hence, if $c'$ is a chain in $C_k(\X',f')$, 
    it is a chain in  $C_k(\X,f)$, too.
    
    Now suppose $c' = (q_s, \ldots,q_1,k)$ 
    is an arbitrary subchain in  $SC_k(\X,f)$,
    then by \cref{not chain satisfied 1},
    none of the bundles in $c'$ are \efxf\ for agent $i$; 
    therefore, since \isbstffor{X_p}{i}{\X}, by \cref{hierarchy}, we get that
    for every $\ell \in [s]$,  $X_p \greaterval{i} X_{q_\ell}$.
    Hence, $c'$ does not contain $p$;
    consequently, $c'$ is also a subchain in  $SC_k(\X',f')$. 
    Also, if $c'$ is a chain in  $C_k(\X,f)$, 
    by $\PA \greaterval{i} X_{q_s}$, we get that
    $c'$ is a chain in   $C_k(\X',f')$, too.
    Therefore, proof is complete.
\end{proof}

\begin{restatable}{lemma}{lemseven}
\label{repeat bstf}
    If an iteration of phase 1 starts with $(\X,f)$ satisfying phase 1's invariant, and
    repeats in \cref{repeat 1.1} with updated input $(\X',f')$, 
    then $(\X',f')$ satisfy the Phase 1's invariant.
\end{restatable}
\begin{proof}
    We have that $\X' = \X$. 
    Also, since \isbstffor{\PA}{i}{\X}, we get \iseflmffor{\PA}{i}{\X}.
    Additionally, we have that before \cref{repeat 1.1}, $i\not\in \range (f)$, so after \cref{repeat 1.1} by setting $f(p)=i$,
    the resulting association function $f'$ is \mefaf\ for $\X'=\X$. 
    Also, since $p\ne k$ (by \cref{p not k proc 1}) and $D(f)=[k-1]$, 
    we get $D(f')= [k-1]$.    

    By \cref{bstf inc}, $C_k(\X',f')=SC_k(\X,f)$ and $C_k(\X',f')=C_k(\X,f)$. 
    Hence, since every subchain $c'\in SC_k(\X,f)$ is contained in a chain $c''\in C_k(\X,f)$,
    every subchain $c'\in SC_k(\X',f')=SC_k(\X,f)$ is contained in a chain $c''\in C_k(\X',f')=C_k(\X,f)$.
\end{proof}

\begin{lemma}
\label{not chain satisfied}
    If an iteration of phase 1 starts with $(\X,f)$ satisfying phase 1's invariant,
    and if \textsc{ReBalance} reaches  \cref{case 1 proc 1}, then 
    \isbstffor{\PA}{i}{\X}, \isbstffor{\PA}{j}{\X},
    and  $|\PA|\geq 2$. 
    Also, for every subchain $\QB \rightsquigarrow X_k$ in $G(\X,f)$
    and for both agents $u \in \{i,j\}$, we have 
    $X_p \setminus x_p^u \greaterval{u} X_{q}$.
\end{lemma}
\begin{proof}
    Since \textsc{ReBalance} has reached \cref{case 1 proc 1}, 
    it has not repeated the while loop in \cref{repeat 1.1};
    hence, $X_p$ is \bstf\ bundle for both agents $i,j$.
    Also, by \cref{not chain satisfied 1}, $X_{q}$ is not \efxf\ for either of agents $i,j$.
    So, by \cref{hierarchy}, we get that 
    $\PA \setminus \xu \greaterval{u} X_{q}$.
    Thus, since $(k)$ is a subchain in $G(\X^{k},f)$, we get
    that  $\PA \setminus \xu \greaterval{u} X_{k}$, 
    so $\PA \setminus \xu \ne \emptyset$, which in turn gives $|\PA| \geq 2$.
\end{proof}


\begin{lemma}
\label{repeat proc 1 case 1 ok}
    If an iteration of phase 1 starts with $(\X,f)$ satisfying phase 1's invariant, and
    repeats in \cref{envy elimination proc 1} with updated input $(\X'',f'')$, 
    then $(\X'',f'')$ satisfy the Phase 1's invariant.
\end{lemma}

\begin{proof}
    Denote partition and association function right after \cref{change in case 1}
    by $(\X',f')$, and right after \cref{envy elimination proc 1} by $(\X'',f'')$.
    We first prove that right after \cref{change in case 1}, 
    for every $\ell \in [k-1]$, $X'_\ell$ is \eflmf\ for agent $f'(\ell)$.
    Since \textsc{ReBalance} repeats in \cref{envy elimination proc 1}, for some $u \in \{i,j\}$ we have
    $\PA \setminus \xu \greatereqval{u} (\QB\cup \xu)$, and by \cref{not chain satisfied}, \isbstffor{\PA}{u}{\X},
    so $\PA \setminus \xu = \argmax_u(\PA \setminus \xu , \QB\cup \xu)$
    is \efxf, therefore \eflmf, for agent $f'(p)=u$ in $\X'$ by \cref{EFX-Best prop to remove good}.

    By \cref{not chain satisfied}, $|X_p|\geq 2$.
    Also, since $c$ is a chain (as defined in \cref{choose chain:proc 1 alg}), $\QB$ is a source.
    Consequently, for every  $\ell \in [k]\setminus \{p\}$ we have that  $X_\ell \greatereqval{f(\ell)} \QB$.
    Additionally, by Phase 1's invariant, $f$ is an \mefaf\ for $\X$, so \iseflmffor{X_\ell}{f(\ell)}{\X}.
    Hence, by by \cref{EFL two} case 1, by removing $\xu$ from $\PA$, and adding it to bundle $\QB$, leading to partition $\xp$, 
    we get that for  every $\ell \in [k]\setminus p$, \iseflmffor{X'_\ell}{f(\ell)}{\xp}.
     By $f(\ell)=f'(\ell)$, we get  \iseflmffor{X'_\ell}{f'(\ell)}{\xp}.
    Hence, $f'$ is an \mefaf\ for $\X'$.

    Since exactly before \cref{envy elimination proc 1}, 
    we have that $f'$ is an \mefaf\ for $\X'$, by \cref{envy cycle elimination ends},
    this property holds after it, too, i.e.,\ $f''$ is an \mefaf\ for $\X''$.
    In addition, since by \cref{envy elimination proc 1}, 
    $G(\X'',f'')$ has no cycles, by \cref{chain existence}, every subchain has a chain containing that subchain.
    Additionally, since $D(f)=[k-1]$, and $p\ne k$, we get that $D(f'')=D(f')=D(f)=[k-1]$.
    So, the proof is complete.
\end{proof}

\subsection{Chain dominance and Phase 1 Termination}
\label{chain dominance main body}
    In this section, we will first introduce the concept of chain dominance and then  
    prove that the first while loop of \textsc{ReBalance} repeats finitely many times.

\begin{definition}[Chain Dominance]
\label{chain dominance}
    Let $\X$ and $\xp$ be two partitions of size $k$, and $f$, $f'$ be two association functions. 
    If $c' = (q'_{s'},\ldots,q'_1, k)$ is a chain in $C_k(\xp,f')$ and $c = (q_{s},\ldots,q_1, k)$ is a chain in $\ck(\X,f)$,
    we say that $c'$ in $\ck(\xp,f')$ \emph{strictly dominates} $c$ in $\ck(\X,f)$, 
    if  $X_{k} \subseteq X'_{k}$
    and at least one of the following holds:
    \begin{enumerate}
        \item   $s'<s$ and for every $\ell \in [s']$, we have 
                $X'_{q'_\ell} = X_{q_\ell}$ and
                $f'(q'_\ell) = f(q_\ell)$.
				
        \item   There exists some $b \in [s]$ such that for every $\ell \in [b-1]$:
    			$X'_{q'_\ell} = X_{q_\ell}$ and $f'(q'_\ell) = f(q_\ell)$. 
     			Also, we have $f'(q'_b) = f(q_b)$ 
                and at least one of the
                $X_{q_b} \lowerval{f(q_b)} X'_{q'_b} $ or 
				$X_{q_b} \subset X'_{q'_b}$ holds. 
    
        \item   $X_{k} \subset X'_{k}$
    \end{enumerate}
    When $(\xp, f')$ and $(\X, f)$ are clear from the context, we just say $c'$ strictly dominates $c$ and denote it as $c'\gg c$. We say $c'$ in $(\xp,f')$ weakly dominates $c$ in $(\X,f)$, denoted as $c'\ggeq c$, if either $c'\gg c$ or we have that 
    $c'=c$ (see \cref{subchain equality} for a definition of $c=c'$).
\end{definition}


\begin{definition}
    We say that some allocation $(\xp,f')$ \emph{weakly dominates} another allocation $(\X,f)$, denoted  $(\xp,f')\ggeq (\X,f)$,
    if for every chain $c'$ in $\ck(\xp,f')$ there exists some chain $c$ in $\ck(\X,f)$ such that $c' \ggeq c$. 
    We also say that $(\xp,f')$ \emph{strictly dominates} $(\X,f)$, denoted  $(\xp,f')\gg (\X,f)$,
    if $(\xp,f')\ggeq (\X,f)$, 
    and one of the following holds:
    \begin{enumerate}
        \item   there exists a chain $c'\in C_k(\xp,f')$ and a chain $c\in C_k(\X,f)$ such that $c' \gg c$, or 
        \item  $C_k(\xp, f') \subset C_k(\X, f)$
        (see \cref{subset chain equality} for a definition).
    \end{enumerate}
\end{definition}

    We defer the proof of \cref{strict dominance proceeds} and \cref{strict dominance different chains}  to \cref{chain dominance appendix}.

\begin{restatable}{lemma}{transitivityXf}
\label{strict dominance proceeds}
    The dominance relation is transitive, i.e., if $(\xz,f'') \ggeq (\xp,f')$ and $(\xp,f')\ggeq (\X,f)$, then $(\X'',f'')\ggeq (\X,f)$.
    Also, if at least one of the dominances is strict, we get $(\X'',f'')\gg (\X,f)$.
    \end{restatable}

\begin{restatable}{lemma}{dominancechange}
\label{strict dominance different chains}
    Let $\X, \xp$ be two partitions and $f,f'$ be two association functions. If $(\xp,f')\gg (\X,f)$, then $(X',f')\ne (X, f)$.    
\end{restatable}


\begin{lemma}
\label{strict dominance of proc 1}
    If the first while loop in \textsc{ReBalance} starts with $(\X,f)$ 
    and repeats itself in \cref{envy elimination proc 1} 
    with input $(\X',f')$, then 
    $(\X',f') \gg (\X,f)$.
\end{lemma}

\begin{proof}    
    We first show that for every chain $c'\in \ckxp$, there exists a chain $c\in \ckx$ such that $c'$ weakly dominates $c$.
    Since any bundle except $\PA$ does not lose any good in transition from $\X$ to $\xp$, 
    and $p\ne k$ (by \cref{p not k proc 1}), 
    we get that $X_{k} \subseteq X'_{k}$.
    
    Let $c' =(q'_{s'}, \ldots,q'_1,k)$  be an arbitrary chain in $\ckxp$.
    Let $z \in \{0,1,\ldots, s'\}$ be the greatest integer for which there exists  some subchain 
    $c'' = ( q_z, \ldots,q_1, k)$  
    in $SC(\X,f)$     
    such that for every $1\leq h \leq z$, we have
    $X_{q_h} = X'_{q'_h}$ and $f(q_h) = f'(q'_h)$. 
    Note that this $z$ exists because for $z=0$, $c''=(k)$ has the mentioned property.
    Since $c''$ is a subchain in $SC_k(\X,f)$, 
    and since $(\X,f)$ satisfies Phase 1's invariant by \cref{phase 1 inavariant holds}, 
    there exists a chain $c^*$ with length $s$ in $\ckx$ such that $c''$ is a subchain of $c^*$, so $s \geq z$. 

    If $X_{k} \ne X'_{k}$, since we have $X_{k} \subseteq X'_{k}$, we get 
    $X_{k} \subset X'_{k}$. Thus, since $c^*$ is a chain in $\ckx$, 
    then $c'$ dominates $c^*$, since the third condition of \cref{chain dominance} holds.
    Hence, assume $X_{k} = X'_{k}$.    
    We have that $z \leq  s'$. We consider two cases.

    \noindent\textbf{$\bullet$ \CaseA. $\mathbf{z=s'}$:}
    In this case, we get $c' = c''$.
    If $c^* = c''$, then $c^* = c'$, so $c'$ weakly dominates $c^*$, 
    so assume $c^*\ne c''$. Thus, $s' =z < s$,  so $c'$ strictly dominates $c^*$,
    because the first condition of \cref{chain dominance} holds.

    \noindent\textbf{$\bullet$ \CaseB. $\mathbf{z<s'}$:}
    Since $c'$ is a chain, we  have that $X'_{q'_{z+1}} \lowerval{f'(q'_{z+1})} X'_{q'_{z}}$.
    Also, we have $X'_{q'_{z}}=X_{q_{z}}$.
    Therefore, we get that $X'_{q'_{z+1}} \lowerval{f'(q'_{z+1})} X'_{q'_{z}}=X_{q_{z}}$.
    
    Let $u\in \{i,j\}$ be the agent defined in \cref{case 1 proc 1}. 
    Since $c''$ is a subchain in $\sckx$ and \textsc{ReBalance} has repeated itself in \cref{envy elimination proc 1},  
    by \cref{not chain satisfied},
    we have that $\PA \setminus \xu \greaterval{u}   X_{q_{z}}$. 
    Also, the bundle associated with agent $u$ right
    after \cref{change in case 1}, is $\PA \setminus \xu$,
    and since \cref{envy elimination proc 1} is envy cycle elimination, after that
    agent $u$ has a bundle with a value not lower than $\PA \setminus \xu$ w.r.t.\ to agent $u$, which
    is higher than the value of $X_{q_{z}}$ w.r.t.\ agent $u$. 
    Therefore, since $X'_{q'_{z+1}} \lowerval{f'(q'_{z+1})} X_{q_{z}}$, 
    we get that $f'(q'_{z+1}) \ne u$.

    Let the bundle associated with agent $f'(q'_{z+1})$ in the partition $\X$ to be $X_t$, 
    i,e., $f(t) =f'(q'_{z+1})$. 
    Since $f'(q'_{z+1}) \ne u$, $t \ne p$, 
    we get that bundle $X_t$ does not lose any good in transitioning from $\X$
    to $\X'$. 
    Since bundle $X'_{q'_{z+1}}$ is the bundle that agent $f'(q'_{z+1})$
    will we be associated with after envy cycle elimination, 
    $X_t \lowereqval{f'(q'_{z+1})} X'_{q'_{p+1}} $. 
    Combining this with $X'_{q'_{z+1}} \lowerval{f'(q'_{z+1})} X_{q_{z}}$, we 
    get that $X_t \lowerval{f'(q'_{z+1})} X_{q_{z}}$.
    
    Hence, $c'''=  (t,q_{z}, \ldots,q_1, k) $ 
    is a subchain in $\sckx$, 
    and by Phase 1's invariant, there exists a chain $\tilde{c}$ in $\ckx$ that contains $c'''$.
    If $X_t = X'_{q'_{z+1}}$, 
    then $c'''$ is a subchain with a length greater than $z$ with the property that we wanted for $c''$, which
    is in contradiction with $z$ being the greatest integer having that property; therefore,
    $X_t \ne X'_{q'_{z+1}}$.
    
    Since $f'(q'_{z+1}) \ne u$, in transition from $\X$ to $\X'$,
    the bundle associated with agent $f'(q'_{z+1})$
    could be changed either by getting a good in \cref{change in case 1}, 
    or getting another bundle in envy cycle elimination of \cref{envy elimination proc 1}, or both.
    If she has got another bundle in the envy cycle elimination, 
    we will have $X_t\lowerval{f(q'_{z+1})} X'_{q'_{z+1}}$; therefore,
    $c'\gg \tilde{c}$, since the second condition of \cref{chain dominance} holds.
    If her bundle has not been changed in the envy elimination cycle, 
    then she has got a new good in \cref{change in case 1}; thus,
    we have $X_t \subset X'_{q'_{z+1}}$, so
    $c'\gg \tilde{c}$, since the second condition of \cref{chain dominance} holds.
    
    Hence, we have shown that for every chain $c'$ in $\ckxp$, there exists some chain $\tilde{c}$ in $\ckx$ such that $c'\ggeq \tilde{c}$.
    Now we show that either there exists some chain in $\ckxp$ that strictly dominates some chain in $\ckx$, or  $\ckxp \subset \ckx$.
    Suppose that the first one does not hold, then since for every chain $c' \in \ckxp$, there exists some chain $\tilde{c}\in \ckx$ such that 
    $c'\ggeq \tilde{c}$, and since by assumption we have that $c'$ does not strictly dominates $\tilde{c}$, we get $c'=\tilde{c}$, 
    and therefore $\ckxp \subseteq \ckx$.
    Let $c \in \ckx$ be the chain in \cref{choose chain:proc 1 alg}. Then, we show that 
    $c \notin \ckxp$, and we  conclude that $\ckxp \subset \ckx$.
    Since in \cref{change in case 1}, we change bundle $\QB$ to the $\QB \cup \xu$, 
    and since in \cref{envy elimination proc 1} partition does not change, bundle $\QB$ is not a bundle in $\xp$, so $\QB$ cannot appear in any chain in $\ckxp$; 
    thus, $c\notin \ckxp$.
    So, the proof is complete. 
\end{proof}

We can now conclude with the proof of the phase 1 termination.

    \begin{proof}[\textbf{Proof of \cref{thm:phase1terminates}}]
    By \cref{ending:proc 1}, every iteration of first while loop ends and executes properly, so we need to prove that it repeats finite times.
    The first while loop in \textsc{ReBalance} can repeat itself either in \cref{repeat 1.1} or \cref{envy elimination proc 1}.
    We first argue that finite repetition in a row can occur in \cref{repeat 1.1}, and then,
    we show that finite repetition occur in \cref{envy elimination proc 1}.
     
    If first while loop starts with $(\X,f)$ and repeats itself by $(\xp,f')$ in \cref{repeat 1.1},
    then $(\xp,f') \ggeq (\X,f)$ because by \cref{bstf inc}, for every chain $c'$ in $G(\xp,f')$, $c'$ is a chain in $G(\X,f)$, and $c' \ggeq c'$. 

    $f$ and $f'$ differ only in the agent associated with the bundle  
    $X_p$, which is \bstf\ for agent $f'(p)=i$ and not \bstf\ for agent $j=f(p)$  
    in partition $\X' = \X$ (by \cref{X_z not bstf for j:proc 1});
    hence:
    \begin{align}
    |\{\ell\in D(f'): \text{\isbstffor{X'_\ell}{f'(\ell)}{\X'}}\}| = |\{\ell\in D(f): \text{\isbstffor{X_\ell}{f(\ell)}{\X}}\}|+1        
    \end{align}

    Therefore, there may be at most $k$ repeats in a row of \cref{repeat 1.1} because this set can be at most $k$.
    Hence, we need to show that there are finitely many repetitions in \cref{envy elimination proc 1}.

    If first while loop starts with $(\X,f)$ and repeats itself by $(\xp,f')$ in \cref{envy elimination proc 1}, 
    by \cref{strict dominance of proc 1}, $(\xp,f') \ggeq (\X,f)$.
    
    Suppose first while loop has been repeated $t$ times in \cref{envy elimination proc 1},
    and for $j \in [t]$, let $(\X^{(j)},f^j)$ be the partition and its association function exactly
    after the $j$-th repetition in \cref{envy elimination proc 1}.
    Then, by \cref{strict dominance proceeds}, we get that for every $j\in [t-1]$, $(\X^{(j+1)},f^{(j+1)}) \gg (\X^{(j)},f^{(j)})$.

    Using the transitivity of the dominance relation (\cref{strict dominance proceeds}), we get that for every $i > j$, we have $(\X^{(i)},f^{(i)})\gg(\X^{(j)},f^{(j)})$. 
    Using \cref{strict dominance different chains}, for every $i > j$, $(\X^{(i)},f^{(i)})\ne (\X^{(j)},f^{(j)})$. 
    There are finite goods. Therefore, there exists a finite number
    of partitions, and there are finite agents, so there are finite association functions. 
    Thus, there are a finite number of different possible $(\X,f)$'s.
    So, we cannot have an infinite sequence of distinct $(\X^{(j)},f^{(j)})$'s. 
    Hence, first while loop will end eventually. 
\end{proof}

\section{Termination of Phase 2 of \textsc{ReBalance}}\label{sec:phase2}
\label{proc 2 main body}

To prove that Phase 2 of the \textsc{ReBalance} algorithm also always terminates within a finite number of iterations, we use $v_j(\RE)$, i.e., the value of agent $j$ for bundle $\RE$, as a potential function. Specifically, we show the following technical lemma.

\begin{lemma}
\label{E does not drop}
    After each iteration of the Phase 2 while-loop, the value $v_j(\RE)$ weakly increases. Also, in any iteration where $v_j(\RE)$ remains the same, the index $\Ip$ also remains unchanged and the size of bundle $\PA$ strictly decreases.
\end{lemma}

Given this lemma, it is easy to show our main result for this section.

\begin{theorem}
\label{main lemma of proc 2}
     Phase 2 of the \textsc{ReBalance} algorithm terminates after a finite number of iterations.
\end{theorem}
\begin{proof}

    First, note that the value $v_j(\RE)$ can strictly increase only a finite number of times, since the number of distinct bundles, and hence also the number of distinct values for $j$, is finite. Then, using \cref{E does not drop}, we can also conclude that the number of iterations during which $v_j(\RE)$ remains the same is finite as well. Specifically, while $v_j(\RE)$ remains the same, the index $\Ip$ remains the same, so it keeps referring to the same bundle, $\PA$, and the size of that bundle strictly decreases. This can happen at most $m$ times (the total number of goods) before $\PA$ is empty.
    Therefore, the total number of iterations of the Phase 2 while-loop is finite.
\end{proof}
 
To prove \cref{E does not drop}, we will show that the following invariant holds during Phase 2:

\begin{definition}[Invariant A of Phase 2]
    We say $(\X,\RE,j,p)$ satisfy \inva\ of Phase 2 if
        \iseflffor{\RE}{j}{\X}, and \Isnoteefxf{\RE}{j}{\X}.
                Also, agent $j$ does not \efxenvy\ any bundle in $\X\setminus \{X_p\}$ \with\ bundle $\RE$.
\end{definition}

\begin{lemma}
\label{proc 2 invariant lemma}
    At the end of every iteration of phase 2, $(\X,\RE,i,j,p)$ satisfy Phase 2's \inva.
\end{lemma}

\begin{lemma}
\label{A is bstf}
    If  $(\X,\RE,j,p)$ satisfy \inva\ of Phase 2,  
    then \isbstffor{\PA}{j}{\X}, $|\PA|\geq 2$, and $\RE \lowerval{j} \PA\setminus \xj$.
\end{lemma}
\begin{proof} 
    By phase 2 \inva, \Isnoteefxf{\RE}{j}{\X}, so \isnotefxffor{\RE}{j}{\X}.
    Hence, by \cref{hierarchy}, agent $j$ \efxenvies\ her \bstf\ bundle in $\X$ \with\ $\RE$.
    By \inva, agent $j$ does not \efxenvy\ any bundle in $\X \setminus \{\PA\}$ \with\ bundle $\RE$;
    therefore, \isbstffor{\PA}{j}{\X} and $\RE \lowerval{j} \PA\setminus \xj$, and since agent $j$ \efxenvies\ bundle $\PA$ \with\ bundle $\RE$,
    $\PA$ has at least two goods.
\end{proof}

    \textbf{Proof of \cref{E does not drop} using \cref{proc 2 invariant lemma}:}
    \textsc{ReBalance} updates bundle $\RE$ by $Second max_j (\RE,\PA',\SC')$ (in \cref{update E}).
    By \cref{A is bstf}, $\RE \lowerval{j} \PA \setminus \xj = \PA'$, so we get that $Second max_j (\RE,\PA',\SC') \greatereqval{j} \RE$.
    Also, if this inequality is not strict, then $\SC' \lowereqval{j} \RE$, so we get $\SC' \lowerval{j} \PA'$.
    Second while loop changes $p$ only in \cref{PA'<j SC'}, which will not be executed since $\SC' \lowerval{j} \PA'$.
    Therefore, if $v_j(\RE)$ does not increase strictly, $p$ will not change, so because $\PA'=\PA\setminus \xj$, we get $\PA$ will have one less good in the next iteration. 

\subsection{Proof of \cref{proc 2 invariant lemma}}

In order to show that \inva\ holds at the end of every iteration of the Phase 2 while-loop, we show that the following invariant also holds at the end of every iteration.

\begin{definition}[Invariant B of Phase 2]
    We say $(\X,p,q,i)$ satisfy \invb\ of Phase 2 if
    there exists a partition $\Z$ with the same number of bundles as $\X$ and bundles $\Za,\Zb \in \Z$ such that the following properties hold for:
    $(\X,p,q,i,\Z,\Za,\Zb)$:
        \begin{itemize}
            \item \textbf{Property 1:}  $i$ does not \efxenvy\ any bundle in $\Z\setminus\{\Zb\}$ \with\ $\Za$. Also, \Isnoteefxf{\Za}{i}{\Z}.
            \item \textbf{Property 2:} For every $X_\ell \in \X\setminus\{X_p,X_q\}$, at least one of the following hold: 
            
            \emph{(i)} $X_\ell \in \Z \setminus \{\Za,\Zb\},$ \qquad \emph{(ii)} $X_\ell \lowereqval{i} \QB,$~~ or \qquad  \emph{(iii)} $X_\ell \lowerval{i} \Za$
        \end{itemize}
\end{definition}

\begin{lemma}
\label{proc 2 strong invariant lemma}
    Invariants A and B both hold at the end of every iteration of the Phase 2 while-loop.
\end{lemma}
\begin{proof}
    In \cref{proc 2 strong invariant lemma first time}, we will prove that these invariants hold at the end of the first iteration of the second while loop.
    Also, in \cref{E EFL lemma} and \cref{invb holds next},
    we show that if these invariants hold at the end of one iteration, then they also hold at the end of the next iteration.
    Then, the argument will hold by induction.
\end{proof}

\begin{definition}[$(\xp,\REp,p',q')$ notation]
In the remainder of this section, for any given iteration of the Phase 2 while-loop, we use $(\X,\RE,p,q)$ to denote the variables values in \cref{second loop alg} of the pseudocode, and we use $(\xp,\REp,p',q')$ to denote their updated values in \cref{change X proc 2}. 

\end{definition}

\begin{lemma}
\label{f' is ok}
    Suppose $f$ is \mefaf\ for $\X$ with $D(f)=[k]\setminus \{p,q\}$, and $f' =$ \textsc{ShiftSubChain}$(f,c)$, 
    then both $f$ and $f'$ are  \mefaf\ for both $\xp$ and $\xz$ with $D(f')=[k]\setminus\{p,r\}$ and $i,j \notin \range(f')$.
\end{lemma}
\begin{proof}
    Since $f$ is an \mefaf\ for $\X$, $c=\Qc \rightsquigarrow \Qb$, and $f' =$ \textsc{ShiftSubChain}$(f,c)$, by \cref{envy subchain},
    we get that $f'$ is \mefaf\ for $\X$, $D(f')=  (D(f)\cup \{q\})\setminus \{r\} = [k] \setminus \{p,r\}$,
    and $\range(f')\subseteq \range(f)$. So, since $i,j\notin \range(f)$, we get $i,j \notin \range(f')$.
    
    Since $\SC$ is a source in $G(\X,f)$, by \cref{envy subchain}, it is a source in $G(\X,f')$.
    If it is the first time that the second while loop executes, by \cref{not chain satisfied}, we get $|\PA|\geq 2$.
    Otherwise, by induction, we will use the fact that \inva\ holds, so by \cref{A is bstf}, $|\PA|\geq 2$. 
    Hence, since $p\notin D(f)$ and $p\notin D(f')$, by \cref{EFL two} case 1,
    we get that if we remove some good from $\PA$ and add it to $\SC$, both $f$ and $f'$ 
    remain to be \mef association functions for the resulting partition ($\xp$ or $\xz$).
\end{proof}

\begin{restatable}{lemma}{finvariantsecondproc}
\label{f inv}
    If \textsc{ReBalance} reaches \cref{second loop alg}, then $D(f)=[k]\setminus \{p,q\}$, $i,j \notin \range(f)$, $f$ is a \mefaf\ for $\X$, and
    there is no \fmefaf\ for $\X$.
\end{restatable}

\begin{proof}
    We first prove that these properties hold for the fist time that \textsc{ReBalance} executes \cref{second loop alg}.
    By \cref{phase 1 inavariant holds}, when \textsc{Rebalance} is in \cref{first loop alg},
    $(\X,f)$ satisfy Phase 1's invariant, and if the first while loop does not repeat, \textsc{ReBalance} does not change $(\X,f)$,
    so $(\X,f)$ still satisfy Phase 1's invariant at the \cref{break 1}, and we have $i\notin \range(f)$, $f(p)=j$, $\QB \rightsquigarrow X_k$, and $D(f)=[k-1]$.
    So, after \cref{change f before phase 2} (\textsc{ShiftSubChain}$(f,c)$ and $f(p)\gets 0$ ), by \cref{envy subchain}, $f$ is a \mefaf\ for $\X$ and $D(f)=[k]\setminus \{p,q\}$.
    Also, $i,j \notin \range(f)$.
    Additionally, since \textsc{ReBalance} has not returned in \cref{i chain satisfied return proc 1},     there is no \fmefaf\ for $\X$.

    Hence, in order to prove that the mentioned properties hold in the next iterations, it is enough to assume that they hold in one iteration, and
    show that they hold in the next iteration.
    \textsc{ReBalance} never associates $i,j$ with any bundle if it does not return, so  $i,j \notin \range(f)$.
    Also, $(f,p,q)$ may only change in
    \cref{PA'<j SC'} or  \cref{q gets i second max}.
    In both cases, by \cref{envy subchain}, $D(f)=[k]\setminus \{p,r\}= [k]\setminus \{p',q'\}$. 
    Also, there is no \fmefaf\ for $\xp$, since \textsc{ReBalance} has not returned in \cref{i chain satisfied return proc 2 X'}.
    Also, by \cref{f' is ok}, we get that $f$ and $f'$ are both \mefaf\ for $\xp$. 
\end{proof}

\begin{restatable}{lemma}{stillwant}
\label{rest mms}
    Suppose that $\PA$,$\SC$ are two disjoint bundles, $i,j$ are two agents, and
    $(\xii,\xj)$ are the output of \textsc{FindMins}$(\PA,i,j)$.
    Then, if for some agent   $u\in \{i,j\}$, 
    we have
    $\PA \setminus \xu \lowereqval{u}  \SC \cup \xu$, then  for $h\in \{i,j\} \setminus \{u\}$, we  have  that
    $\PA \setminus x_p^h \lowereqval{u} \PA \setminus \xu \lowereqval{u} \SC \cup x_p^h$.
\end{restatable}

\begin{proof}
    If $\xii=\xj$, then we have $\SC \cup x_p^h= \SC\cup \xu \greatereqval{u}  \PA \setminus \xu = \PA \setminus x_p^h$, 
    and if $\xii \ne \xj$, then by the choice of $\xii$ and $\xj$,
    we have $\PA \setminus \xu\greaterval{u} \PA \setminus x_p^h $.
    Since the valuation of agent $u$ is restricted MMS-feasible, we have that 
    $\argmax_u(\PA \setminus x_p^h,\SC\cup x_p^h)  \greatereqval{u} \argmin_u (\PA \setminus \xu, \SC \cup \xu) = \PA \setminus \xu$;
    therefore, 
    we have that $ \SC \cup x_p^h =\argmax_u(\PA \setminus x_p^h,\SC\cup x_p^h) \greatereqval{u} \PA \setminus \xu$.
    So, the proof is complete.
\end{proof}


\begin{lemma}
\label{inv B by i bstf}
    If in some iteration of the Phase 2 while-loop we have $\PA' \lowerval{j} \SC'$, \isbstffor{\PA}{i}{\X}, and \isbstffor{\PA}{j}{\X}, 
    then \invb\ holds at the end of that iteration.
\end{lemma}
\begin{proof}
    Let partition $\xz$ be the partition in \cref{construct part X''}, constructed by changing $\X$ by removing $\xii$ from $\PA$
    and adding it to $\SC$ and setting $\PA''=\PA\setminus \xii, \SC'' = \SC\cup \xii$.
    We will show that $(\X',p',q',i,\zp=\xz,\Zap=\PA'', \Zbp=\SC'')$ satisfies \invb. 

    We first show that property 1 of \invb\ holds.
    Since \isbstffor{\PA}{i}{\X}, by \cref{EFX-Best prop to remove good}, 
    agent $i$ does not \efxenvy\ any bundle in $\xz\setminus\{\SC''\}$ \with\ bundle $\PA''=\PA \setminus \xii$. 
    Next, we show that \Isnoteefxf{\PA''}{i}{\xz}.
    
    By \cref{Xr not efxf}, \isnotefxffor{\SC}{i}{\X};
    therefore, by \cref{hierarchy}, $\SC \lowerval{i} \PA \setminus \xii$, 
    and since $\xii\lowereqval{i} \PA\setminus \xii$, 
    we get \iseflffor{\PA \setminus \xii}{i}{\X''}.

    We have $\SC\cup \xj \greaterval{j}  \PA\setminus \xj$,
    and $\xii,\xj=$\textsc{FindMins}$(\Qa,i,j)$ in \cref{choice of xzu proc2}. 
    Hence, by \cref{rest mms}, $\SC\cup \xii \greatereqval{j} \PA \setminus \xj \greatereqval{j} \PA \setminus \xii$.
    We have $\xp \setminus \{\PA',\SC'\} = \xz \setminus \{\PA'',\SC''\}$, and by \cref{EFX-Best prop to remove good},
    agent $j$ does not \efxenvy\ any bundle in $\xp \setminus \{\SC\cup \xj\}$ \with\ $\PA \setminus \xj$ .
    Hence, by $\SC\cup \xii \greatereqval{j} \PA \setminus \xj \greatereqval{j} \PA \setminus \xii$, agent $j$
     does not \efxenvy\ any bundle in $\xz$ \with\ $\SC\cup \xii$.


    Therefore, \isefxffor{\SC\cup \xii}{j}{\xz}. 
    Let $f' =$ \textsc{ShiftSubChain}$(f,c)$, 
    then by \cref{f' is ok},
    $f'$ is an \mefaf\ for $\xz$ with $D(f')=[k]\setminus\{p,r\}$ and $i,j \notin \range(f')$.
    So, if \iseefxffor{\PA''}{i}{\xz}, then \iseflmffor{\PA''}{i}{\xz} and by setting $f'(r)=j$ and $f'(p)=i$,
    we get that $f'$ is a \fmefaf\ for $\xz$, which is a contradiction since \textsc{ReBalance} has not returned in \cref{i chain satisfied return proc 2 X''}.

    Next, we show that property 2 of \invb\ holds.
    Since $\PA' \lowerval{j} \SC'$, by \cref{PA'<j SC'}, we have $p'=r,q'=p$.
    $\xp$ and $\xz$ are both made from changing bundles $\PA,\SC \in \X$, so 
    $\xp \setminus \{\PAp,\QBp\}= \xp \setminus \{\SC',\PA'\} = \xz \setminus  \{\SC'',\PA''\} = \zp \setminus \{\Zap,\Zbp\}$.
    Hence, for every $X_\ell' \in \xp \setminus \{\PAp,\QBp\}$, we have $X_\ell' \in \zp \setminus \{\Zap,\Zbp\}$.
    So, the proof is complete.
\end{proof}

\begin{lemma}
\label{proc 2 strong invariant lemma first time}
    At the end of the first iteration of the second while loop of \textsc{ReBalance},  $(\X,\RE,i,j,p,q)$ satisfy Phase 2's \inva\ and \invb.
\end{lemma}
\begin{proof}
    Since it is the first time that \textsc{ReBalance} executes \cref{change X proc 2}, 
    variables come from Phase 1, so since Phase 1 has not repeated itself in the last iteration,
    by  \cref{E def first}, \cref{not chain satisfied}, and the fact that first while loop has not repeated itself in \cref{envy elimination proc 1} in the last iteration,
    we get that
    in \cref{second loop alg},
    \isbstffor{\PA}{i}{\X},  $\Qb \cup \xii \greaterval{i} \Qa \setminus \xii$, \isbstffor{\PA}{j}{\X}, $\Qb \cup \xj \greaterval{j} \Qa \setminus \xj$,
    $X_q$ is a source in $G(\X,f)$, $|\PA|\geq 2$, and
    $\RE=\PA\setminus \xj$.

    Since $X_q$ is a source in $G(\X,f)$, we get that in \cref{choose c proc 2}, $\SC = \QB$.
    Hence, $\SC \cup \xii \greaterval{i}  \PA \setminus \xii$ and $\SC \cup \xj \greaterval{j}  \PA \setminus \xj$.
    Therefore, by \cref{inv B by i bstf}, after \textsc{ReBalance} executes \cref{change X proc 2}, \invb\ holds. 

    Additionally, $\RE= \PA\setminus \xj \in \xp$. 
    Also, by case 2 of \cref{EFL two}, since $\SC$ is a source, \isbstffor{\PA}{j}{\X}, and $|\PA|\geq 2$,
    we get that \iseflffor{\RE= \PA\setminus \xj}{j}{\xp}.
    By $\SC \cup \xj \greaterval{j}  \PA \setminus \xj$ and \cref{PA'<j SC'}, we get $p'=r$.
    By \cref{EFX-Best prop to remove good}, agent $j$ does not \efxenvy\ any bundle in $\xp \setminus \{\SC'\}= \xp \setminus \{\PAp\}$ \with\ $\RE$.

    Next, we show \Isnoteefxf{\RE}{j}{\xp}.
    Suppose on the contrary that \iseefxffor{\RE}{j}{\xp}, so \iseflmffor{\RE= \PA'}{j}{\xp}.
    Also, $\SC\cup \xii \greaterval{i}  \PA\setminus \xii$,
    and $\xii,\xj=$\textsc{FindMins}$(\Qa,i,j)$ in \cref{choice of xzu proc2}. 
    Hence, by \cref{rest mms}, $\SC\cup \xj \greatereqval{i} \PA \setminus \xii \greatereqval{i} \PA \setminus \xj$.
    We have $\xp \setminus \{\PA',\SC'\} = \xz \setminus \{\PA'',\SC''\}$, and by \cref{EFX-Best prop to remove good},
    agent $i$ does not \efxenvy\ any bundle in $\xz \setminus \{\SC\cup \xii\}$ \with\ $\PA \setminus \xii$ .
    Hence, by $\SC\cup \xj \greatereqval{i} \PA \setminus \xii \greatereqval{i} \PA \setminus \xj$, agent $i$
     does not \efxenvy\ any bundle in $\xp$ \with\ $\SC\cup \xj$.

    Therefore, \isefxffor{\SC\cup \xj}{i}{\xp}. 
    Let $f' =$ \textsc{ShiftSubChain}$(f,c)$, 
    then by \cref{f' is ok},
    $f'$ is an \mefaf\ for $\xp$ with $D(f')=[k]\setminus\{p,r\}$ and $i,j \notin \range(f')$.
    So, by setting $f'(r)=i$ and $f'(p)=j$,
    we get that $f'$ is a \fmefaf\ for $\xp$, which is a contradiction since \textsc{ReBalance} has not returned in \cref{i chain satisfied return proc 2 X'}.
    So, the proof is complete.
\end{proof}

    In the remainder of this section,
    We assume that $(\X,\RE,p,q,i,j)$ satisfy \inva\ and \invb, and we will show that these hold for $(\xp,\REp,p',q',i,j)$, too.
    Also, let $\Z, \Za,\Zb$ be the partition and bundles that exist by \invb\ for $(\X,p,q,i)$.
    We will use \inva\ and \invb\ to prove (provided at the end of the section) the following lemma, 
    and then prove that these invariants hold in the next iteration.



\begin{restatable}{lemma}{maxisefxffori}
\label{max efxf for i}
    If in some iteration of the Phase 2 while-loop $(\X,p,q,i)$ satisfy \invb,  
    then \isefxffor{\PA}{i}{\X} and \isefxffor{\argmax_i(\QB',\PA',\SC')}{i}{\xp}.
\end{restatable}

\begin{restatable}{lemma}{thesame}
\label{the same}
     If in some iteration of the Phase 2 while-loop $(\X,\RE,p,q,i,j)$ satisfy \inva\ and \invb, and \textsc{ReBalance} does not return in this iteration, 
     then \isefxffor{\argmax_i(\PA',\SC')}{i}{\xp}, \isefxffor{\argmax_j(\PA',\SC')}{j}{\xp} and $\argmax_i(\PA',\SC')=\argmax_j(\PA',\SC')$.       
\end{restatable}
\begin{proof}
    We first  prove  \isefxffor{\argmax_i(\PA',\SC')}{i}{\xp}. By \cref{max efxf for i}, \isefxffor{\argmax_i(\QB',\PA',\SC')}{i}{\xp}, so if $q=r$, we are done.
    So suppose not, then $\QB'=\QB$. If $\argmax_i(\QB',\PA',\SC')= \QB'$, then \isefxffor{\QB'}{i}{\xp}, 
    so we get that \ismxsffor{\QB}{i}{\X}, since $\X$ and $\xp$ have the same number of bundles.
    Also, we get that agent $i$ does not \efxenvy\ any bundle in $\X\setminus \{\PA,\SC\}$ \with\ bundle $\QB$.
    Also, because $\QB \greaterval{i} (\PA \setminus \xj, \SC\cup \xj)$, and since valuations are monotone,
    we get that $\QB\greaterval{i} \SC$, $\QB\greaterval{i} \xj$ and $\QB\greaterval{i} (\PA \setminus \xj)$;
    therefore, agent $i$ \with\ bundle $\QB$ does not \eflenvy\ any of the bundles $\SC$ and $\PA$, 
    so \iseflffor{\QB}{i}{\X}. Hence, \iseflffor{\QB}{i}{\X}, which is a contradiction by \cref{Xr not efxf}.
    Hence, $\argmax_i(\QB',\PA',\SC')\ne \QB'$, so \isefxffor{\argmax_i(\PA',\SC')}{i}{\xp}.

    By \cref{E does not drop}, \isbstffor{\PA}{j}{\X}, so by \cref{EFX-Best prop to remove good}, we get \isefxffor{\argmax_j(\PA',\SC')}{j}{\xp}. 

    Next, we show if the second while loop repeats itself, $\argmax_i(\PA',\SC')=\argmax_j(\PA',\SC')$.    
    Let $f' =$ \textsc{ShiftSubChain}$(f,c)$, then by \cref{f' is ok}, we get that
    $f'$ is an \mefaf\ for $\xp$ with $D(f')=[k]\setminus\{p,r\}$ and $i,j \notin \range(f')$.
    If $\argmax_i(\PA',\SC')\ne\argmax_j(\PA',\SC')$, then by associating  $\argmax_i(\PA',\SC')$ with agent $i$ and
    associating $\argmax_j(\PA',\SC')$ with agent $j$, we get a \fmefaf\ for $\xp$, whcih is a contradiction, since \textsc{ReBalance} has
    not returned in \cref{i chain satisfied return proc 2 X'}.
    So, the proof is complete.
\end{proof}

\begin{restatable}{lemma}{sourceisnotgood}
    \label{Xr not efxf}
    In every iteration of the second while loop, for both $u\in \{i,j\}$, \isnotefxffor{\SC}{u}{\X} and \Isnoteflmf{\QB}{i}{\X}.
\end{restatable}
 
\begin{proof}
    Let $c$ be the chain from $\SC$ to $\PA$.
    By \cref{A is bstf} and \cref{max efxf for i}, $\PA$ is \eflmf\ for both agents $i$ and $j$ in $\X$.
    Also, by \cref{f inv}, 
    $f$ is an \mefaf\ for $\X$ with $D(f)=[k]\setminus\{p,q\}$ and $i,j \notin \range(f)$.
    Let $f' =$ \textsc{ShiftSubChain}$(f,c)$, then by \cref{envy subchain}, 
    $f'$ is an \mefaf\ for $\X$ with $D(f')=[k]\setminus\{p,r\}$ and $i,j \notin \range(f')$.
    Therefore, if $\SC$ is \efxf\ for any agent $u\in \{i,j\}$, 
    by setting $f'(r)$ to $u$
    and $f'(p)$ to the remaining agent among $\{i,j\}$, $f'$ would be a \fmefaf\ for $\X$,
    which is a contradiction with the fact that $\X$ does not have a \fmefaf\ by \cref{f inv}.

    Also, by \cref{E does not drop} and \cref{f inv}, \isbstffor{\PA}{j}{\X}, and 
    $f$ is \mefaf\ for $\X$ with $D(f)=[k]\setminus \{p,q\}$, so
    if \iseflmffor{\QB}{i}{\X}, then by setting $f(p)=j$ and $f(q)=i$, we get a \fmefaf\ for $\X$, 
    which is a contradiction by \cref{f inv}.
\end{proof}



\begin{lemma}
\label{E EFL lemma} 
    If Invariants A and B hold at the end of some iteration, and if second while loop does not return in the next iteration,
    then Invariant A holds at the end of the next iteration.
\end{lemma}
\begin{proof}
    We assume \inva\ holds for $(\X,\RE,p,j)$, and we show it holds for the next iteration.
    Let $\RE$ be the bundle in \cref{second loop alg}, and $\REp=\textsc{NextBest}_j (\RE,\PA',\SC')$ be its updated bundle in \cref{update E}. 
    We first show that $\REp \in \xp$.
    We have that $\RE\in \X$, and $\RE \ne \PA$ by \cref{A is bstf}, so if $\RE \ne \SC$, 
    bundle $\RE$ does not change in the transition from $\X$ to $\xp$,
    so $\RE \in \xp$; therefore, $\REp \in \xp$.
    Also, if $\RE=\SC$, since $\SC' = \SC\cup \xj=\RE\cup \xj$ and $\RE\lowerval{j} \PA'$ (\cref{A is bstf}), we get that $\REp\ne \RE$, so $\REp \in \xp$.

    Next, we show that agent $j$ does not \efxenvy\ any bundle in $\X' \setminus\{\PAp\}$ \with\ bundle $\REp$.
    By \inva\ of $(\X,\RE,p,j)$, agent $j$ does not \efxenvy\ any bundle in $\X\setminus\{\PA\}$ \with\ bundle $\RE$.
    By \cref{E does not drop}, $\REp\greatereqval{j} \RE$, and in transition of $\X$ to $\xp$ only bundles $\PA$ and $\SC$ change.
    Hence, since $\REp =\textsc{NextBest}_j (\RE,\PA',\SC')$  and $\PAp = \argmax_j(\PA',\SC')$ (\cref{j gets max}), 
    agent $j$ does not \efxenvy\ any bundle in $\X' \setminus\{\PAp\}$ \with\ bundle $\REp$.
    
    Next, we prove that  \iseflffor{\REp}{j}{\xp}.
    Since agent $j$ does not \efxenvy\ any bundle in $\X' \setminus\{\PAp\}$ \with\ bundle $\REp$, we only need to show that
    agent $j$ does not \eflenvy\ bundle $\PAp$ \with\ bundle $\REp$.
    We divide to two cases:
    
    \textbf{Case 1: $\mathbf{\PA'\greatereqval{j} \SC'}$:} 
    In this case, $\PAp= \PA\setminus \xj \subset \PA$, and by \cref{E does not drop}, $\REp \greatereqval{j} \RE$.
    Since \iseflffor{\RE}{j}{\X}, so agent $j$ does not \eflenvy\ bundle $\PA$ \with\ bundle $\RE$. 
    Hence, agent $j$ does not \eflenvy\ bundle $\PAp$ \with\ bundle $\REp$.

     \textbf{Case 2: $\mathbf{\PA'\lowerval{j} \SC'}$:} 
        In this case, $\PAp= \SC \cup \xj$, and $\REp = \PA \setminus \xj$.
        By \cref{Xr not efxf}, \isnotefxffor{\SC}{j}{\X},
        and since \isbstffor{\PA}{j}{\X}, 
        we get $\SC\lowerval{j} \PA \setminus \xj$. 
        Also, we have $\xj \lowereqval{j} \PA \setminus \xj$;
        therefore, agent $j$ does not \eflenvy\ bundle $\SC\cup \xj$ \with\ $\PA \setminus \xj$, so \iseflffor{\REp}{j}{\xp}.

    Next, we show that \Isnoteefxf{\REp}{j}{\xp}.
    We have that  $\REp =\textsc{NextBest}_j (\RE,\PA',\SC')$.
    If $\REp = \RE$, we get \Isnoteefxf{\REp}{j}{\X'} since \Isnoteefxf{\RE}{j}{\X} and the number of bundles in $\X$
    and $\xp$ is the same. So assume otherwise, i.e., $\REp\ne \RE$.
    Hence, since $\PA' \greatereqval{j} \RE$ (by \cref{E does not drop}), we get $\REp = \argmin_j(\PA',\SC')$.
    If \Isnoteefxf{\REp}{j}{\xp}, we are done. so assume otherwise: \iseflmffor{\REp=\argmin_j(\PA',\SC')}{j}{\xp}.    
    Also, by \cref{max efxf for i} and \cref{the same}, \isefxffor{\argmax_i(\PA',\SC')=\argmax_j(\PA',\SC')}{i}{\xp}.
    
    Let $f' =$ \textsc{ShiftSubChain}$(f,c)$, then by \cref{f' is ok},
    $f'$ is an \mefaf\ for $\xp$ with $D(f')=[k]\setminus\{p,r\}$ and $i,j \notin \range(f')$.
    Therefore, by associating bundle $\argmax_j(\PA',\SC')$ with agent $i$ and associating bundle  $\argmin_j(\PA',\SC')$ with agent $j$, 
    $f'$ would be a \fmefaf\ for $\xp$,
    which is a contradiction with the fact that $\xp$ does not have a \fmefaf\ since \textsc{ReBalance}  has not returned in \cref{i chain satisfied return proc 2 X'}.
\end{proof}

\begin{lemma}
\label{invb holds next}
    If Invariants A and B hold at the end of some iteration, and if second while loop does not return in the next iteration,
    then \invb\ holds in the next iteration.
\end{lemma}
\begin{proof}
    We divide to two cases.

    \textbf{Case 1. $\PA' \greatereqval{j} \SC'$ or $\QB \lowerval{i} \Za$:}
    In this case, we show that $(\X',p',q',i,\Z,\Za,\Zb)$ satisfy \invb.
    Property 1 of \invb\ holds in the next iteration because it only involves partition $\Z$, and we are not changing it.
    We now prove that property 2 of \invb\ holds, too.
    Note that in every case of second while loop, $p \in \{p',q'\}$, 
    and in the transition from $\X$ to $\xp$, only bundles $\PA$ and $\SC$ change, so
    for every $X_\ell' \in \xp \setminus \{\PAp,\QBp\}$, 
    we have  that either  $X'_\ell \in \{\SC',\QB\}$ or $X_\ell'\in \X \setminus \{\PA,\QB\}$.

    \textbf{Subcase 1.1 $\PA' \greatereqval{j} \SC'$:}
        In this case, \cref{i gets second max} is executed, so
        by \cref{q gets i second max}, we have $\QBp=\argmax_i(\SC', \QB )$, so if $X'_\ell \in \{\SC',\QB\}$, we get
        $X_\ell' \lowereqval{i}  \argmax_i(\SC', \QB ) = \QBp$. 
        Also, if $X_\ell'\in \X \setminus \{\PA,\QB\}$,
        since property 2 of \invb\ holds for $(\X,\Z)$, 
        we get that either $X_\ell'\in \Z \setminus \{\Za,\Zb\}$, or $X_\ell' \lowereqval{i} \QB \lowereqval{i} \QBp$,
        or $X_\ell' \lowerval{i} \Za$.

    \textbf{Subase 1.2. $\PA' \lowerval{j} \SC'$ and $\QB \lowerval{i} \Za$:}
        Since $\PA' \lowerval{j} \SC'$, by \cref{PA'<j SC'}, we get that $p'=r$.
        So by $X_\ell' \in \xp \setminus \{\PAp,\QBp\}$, we get  $X'_\ell \ne \SC'$.
        Hence, if $X'_\ell \in \{\SC',\QB\}$, we get $X'_\ell = \QB \lowerval{i} \Za$.
        Also, if $X_\ell'\in \X \setminus \{\PA,\QB\}$,
        since property 2 of \invb\ holds for $(\X,\Z)$, 
        we get that either $X_\ell'\in \Z \setminus \{\Za,\Zb\}$, or $X_\ell' \lowereqval{i} \QB \lowerval{i} \Za$,
        or $X_\ell' \lowerval{i} \Za$.

        \textbf{Case 2.  $\PA' \lowerval{j} \SC'$ and $\QB \greatereqval{i} \Za$:}
            By \cref{A is bstf}, we have \isbstffor{\PA}{j}{\X}, and we have $\PA' \lowerval{j} \SC'$, so
            by \cref{inv B by i bstf}, we only need to show that \isbstffor{\PA}{i}{\X}.
            We first show that agent $i$ does not \efxenvy\ any bundle in $\X \setminus \{\PA\}$ \with\ $\QB$.
            By property 2 of \invb\ for $(\X,\Z)$, for every $X_\ell \in \X \setminus \{\PA\}$, we have that
            either $X_\ell \in \Z \setminus \{\Za,\Zb\}$,
            or $X_\ell \lowereqval{i} \QB$, or $X_\ell \lowerval{i} \Za$.
            In the first case, by property 1 of \invb, agent $i$ does not \efxenvy\ bundle $X_\ell$ \with\ bundle $\Za$, 
            so since $\QB \greatereqval{i} \Za$, agent $i$ does not \efxenvy\ $X_\ell$ \with\ bundle $\QB$, too.
            Also, in the second and third cases, since $\QB \greatereqval{i} \Za$, we will have that $X_\ell \lowereqval{i} \QB$, 
            so agent $i$ does not \efxenvy\ $X_\ell$ \with\ $\QB$ in these cases, too.
            Hence, agent $i$ does not \efxenvy\ any bundle in $\X \setminus \{\PA\}$ \with\ bundle $\QB$.
            By \cref{Xr not efxf}, \isnotefxffor{\QB}{i}{\X} and since agent $i$ does not \efxenvy\ any bundle in $\X \setminus \{\PA\}$ \with\ bundle $\QB$, 
            by \cref{hierarchy}, we get that \isbstffor{\PA}{i}{\X}.
\end{proof}





    Now, it only remains to prove \cref{max efxf for i} assuming that \inva\ and \invb\ hold.
    We first need the following lemma.
    In the following lemma, partition $\X$ and $\Z$ are arbitrary partitions 
    and not necessarily partitions $\X$ and $\Z$ in the second while loop.
    
\begin{lemma}
\label{max efxf}
    Suppose that $\Z$ and $\X$ are two partitions with the same number of bundles and $\Za \in \Z$, 
    $\yp \subseteq \X \cap \Z$ is a subset of bundles that exist both in $\Z$ and $\X$ 
    such that agent $i$ does not \efxenvy\ any of the bundles in $\yp$ \with\ $\Za$.
    Also, let  $\Y= \X \setminus \yp$.
    If \isnoteefxffor{\Za}{i}{\Z}, 
    then \isefxffor{\argmax_i(\Y)}{i}{\X}, and $\argmax_i(\Y) \greatereqval{i} \Za$.
\end{lemma}
\begin{proof}
    Consider the following process. We start with the partition $\Z$. 
    As long as there exist $Z_p,Z_q \in \Z \setminus \Y'$ such that 
    $Z_{p}\lowerval{i} \Za\lowerval{i} Z_{q} \setminus z_{q}^i$,
    remove $z_{q}^i$ from $Z_{q}$ and add it to the $Z_{p}$.
    
    After each step of this process, 
    either the number of bundles with a lower value than bundle $\Za$ w.r.t.\ agent $i$ decreases
    or it does not change, which in the latter, 
    the number of goods in the bundles with a lower value than bundle $\Za$ w.r.t.\ agent $i$ increases. 
    So, there could be at most $m$ repeats in a row wherein the second case occurs since there are $m$ goods.
    Also, since there are $k$ bundles, there could be at most $k$ repeats wherein the first case occurs.
    So, this process will eventually end.

    Since during this process, we do not change bundle $\Za$,
    and since $\Za$ is not \eefxf\ for agent $i$,
    at any point after this process, there exists some $Z_q$ such that 
    $\Za\lowerval{i} Z_{q} \setminus z_q^i$. 
    Since during the process, bundles in $\yp$ do not change, 
    and since agent $i$ does not \efxenvy\ any of the bundles in $\yp$ 
    \with\ bundle $\Za$ (according to Lemma's assumption),
    we get that there exists some $Z_q \in \Z \setminus \yp$ such that 
    $\Za\lowerval{i} Z_{q} \setminus z_q^i$. 
    Hence, after the termination of this process, for every
    $Z_p \in \Z \setminus \yp$, we have that $\Za \lowereqval{i} Z_p$, so $\Za \lowereqval{i} \argmin_i(\Z\setminus \yp)$.
    
    Note that the set of bundles $\Z\setminus \yp$ is another partition of the goods in
    the set of bundles $\Y=\X \setminus\yp$ (with the same number of bundles); therefore, 
    since the valuation function of the agent $i$ is restricted-MMS-feasible, we get that 
    $\argmin_i(\Z\setminus \yp) \lowereqval{i} \argmax_i(\Y)$; therefore, $\Za \lowereqval{i} \argmax_i(\Y)$. 
    Since agent $i$ did not \efxenvy\ any bundle in the $\yp$ 
    \with\ bundle $\Za$, and  $\Za \lowereqval{i} \argmax_i(\Y)$, 
    he  does not \efxenvy\ bundles in $\yp$ \with\ the bundle $\argmax_i(\Y)$.
    Also, since $\argmax_i(\Y)$ has the maximum value among bundles in $\Y$, 
    he does not \efxenvy\ bundles in $\Y$ \with\ $\argmax_i(\Y)$. Hence, \isefxffor{\argmax_i(\Y)}{i}{\X}.
\end{proof}



\textbf{Proof of \cref{max efxf for i}.}
    Let $\X$ and $\xp$ be the partitions defined in the second while loop, and let $(\Z,\Za,\Zb)$ be the partition and bundles in \invb.
    Let $\yp = (\xp \setminus \{ \PA',\QB',\SC'\}) \cap (\Z \setminus \{\Za,\Zb\})$ and 
    $\Y = \xp \setminus \yp$.
    We have that $\yp \subseteq \X \cap \Z$.
    Also, we have that $\Zb\notin \yp$; therefore, by property 1 of \invb, 
    agent $i$ does not \efxenvy\ any bundle in $\yp$ \with\ $\Za$.
    Also, by property 1 of \invb, we have \Isnoteefxf{\Za}{i}{\Z}.
    Therefore, by \cref{max efxf}, $\argmax_i(\Y)\greatereqval{i} \Za$ and \isefxffor{\argmax_i(\Y)}{i}{\xp}.
    Now we show that $\argmax_i(\Y) \lowereqval{i} \argmax_i(\PA',\QB',\SC')$.

    For every $X'_\ell \in \Y \setminus \{\PA',\QB',\SC'\}$, 
    we have that $X'_\ell \in \X \setminus \{\PA,\QB\}$,
    so by property 2 of \invb, we get that either 
    $X'_\ell \in \Z \setminus \{\Za,\Zb\}$, or $X'_\ell\lowerval{i} \Za$, or $X'_\ell\lowereqval{i} B$.
    The first case cannot happen since $X'_\ell \in \Y$, so it is not in $\yp$.
    Hence, for every $X'_\ell \in \Y \setminus \{\PA',\QB',\SC'\}$, we have either 
    $X'_\ell\lowerval{i} \Za$ or $X'_\ell\lowereqval{i} \QB$.
    If the first case holds, then $X'_\ell\ne \argmax_i(\Y)$ since we have $\argmax_i(\Y)\greatereqval{i} \Za$. 
    So, since $\QB\subseteq \QB'$, $\QB \lowereqval{i} \argmax_i(\QB',\PA',\SC')$, so we get that
    $\argmax_i(\Y)\lowereqval{i} \argmax_i(\QB',\PA',\SC')$.
    Hence,  \isefxffor{\argmax_i(\QB',\PA',\SC')}{i}{\xp}.

    By similar arguments, by setting $\yp = (\X \setminus \{ \PA,\QB\}) \cap (\Z \setminus \{\Za,\Zb\})$, we get  \isefxffor{\argmax_i(\PA,\QB)}{i}{\X}.
    By \cref{Xr not efxf}, \isnotefxffor{\QB}{i}{\X}. Hence, \isefxffor{\PA}{i}{\X}.

\bibliographystyle{plainnat}
\bibliography{mybibliography}

\section{Implications of our Results}\label{sec:implications}

    Although our main results are stated with respect to the EFL and MXS fairness guarantees, each of these has direct implications regarding other well-studied fairness properties. In this section, we discuss and prove some of these implications. In this section, we denote an allocation by $\X=(X_1,\ldots,X_n)$ in which $X_i$ is allocated to
    agent $i$.
    
\paragraph{$\alpha$-MMS, $\alpha$-GMMS, $\alpha$-PMMS.}
    For a subset $S\subseteq M$ of goods and natural number $k$, let 
    $\Pi_{k}(S)$ be the set of all partitions of $S$ to $k$ bundles, then
    the k-maximin share of agent $i$ with respect to $S$ is:
    $\mu_i(k,S) =\max_{\X \in \Pi_k(S)} \min_{X_j \in \X} v_i(X_j)$.
    An allocation $\X$ is called an 
    $\alpha$-MMS ($\alpha$-maximin share) allocation if
    $v_i(X_i) \geq \alpha \mu_i(n,M)$, for every agent $i$. 
    An allocation $\X $ is called an $\alpha$-GMMS ($\alpha$-groupwise maximin share)
    allocation if for every subset of agents $N'\subseteq [n]$ and any agent $i \in N'$, 
    $v_i(X_i) \geq \alpha \mu_i(|N'|,\bigcup_{j\in N'} X_j)$.
    An allocation $\X = (X_1,\ldots,X_n)$ is called an $\alpha$-PMMS ($\alpha$-pairwise maximin share)
    allocation if for every pair of agents $i, j \in [n]$, 
    $v_i(X_i) \geq  \alpha  \mu_i(2,X_i \cup X_j)$.

\paragraph{Approximate EFX ($\alpha$-EFX).}
    An allocation $\X$ is $\alpha$-EFX if 
    for every two agents $i,j$ and for every $g\in X_j$,
    we have $v_i(X_i) \geq \alpha v_i(X_j\setminus g)$.

\begin{lemma}
If all the valuation functions are additive and an allocation $\X=(X_1,\ldots, X_n)$ is
\mxsf, then it is $\frac{4}{7}$-MMS allocation as well.
\end{lemma}

\begin{proof}
    For every agent $i$, there exists partition $\Y=(Y_1,\ldots,Y_n)$ such that $Y_i$ is \efxf\ for agent
    $i$ in $\Y$ and $X_i\greatereqval{i} Y_i$. By \cite{ABM18}, we have that every EFX allocation is 
    $\frac{4}{7}$-MMS; therefore, we have that $X_i\greatereqval{i} Y_i \greatereqval{i} \frac{4}{7} \mu_i(n,M)$.
\end{proof}

    In Lemma 3 of \cite{BBMN18}, 
    it has been proven that every EFL allocation is also a $\frac{1}{2}$-GMMS when valuation functions are additive.
    Next, we prove the rest of the lemma.
    Although the following proof is not particularly demanding, to the best of our knowledge, 
    this implication has not been shown in any prior work, so we provide the proof below.

\begin{lemma}
    If all the valuation functions are additive and an allocation $\X=(X_1,\ldots, X_n)$ is
    \eflf, then it is $\frac{1}{2}$-GMMS, $\frac{1}{2}$-EFX and $\frac{2}{3}$-PMMS allocation as well.
\end{lemma}
\begin{proof}

    Suppose that allocation $\X$ is EFL, then for every agent $i$ with bundle $X_i$ 
    and another bundle $X_j$, one of the following holds:
    \begin{itemize}
        \item case 1: $|X_j|\leq 1$,
        \item case 2: there exists $g\in X_j$ such that 
        $v_i(X_i) \geq v_i(\{g\})$ and $v_i(X_i)\geq v_i(X_j \setminus\{g\})$.
    \end{itemize}
    In the first case, after removing any good from $X_j$, it will turn into an empty set,
    so in the first case, there is no \efxenvy\ towards bundle $X_j$.
    Also, in the second case, we get that 
    $2v_i(X_i)\geq v_i(\{g\}) +v_i(X_j \setminus\{g\})=v_i(X_j)$,
    so $v_i(X_i) \geq \frac{1}{2}v_i(X_j)$. Therefore, EFL allocations are $\frac{1}{2}$-EFX, too.
 
    Now we show that they are $\frac{2}{3}$-PMMS, too.
    We first repeat Lemma 3.4 of \cite{AMNS15}, which states that
    for any agent $i$, any subset of goods $S$,
    and any natural number $k$, and any good $h \in S$, we have that 
    $\mu_i(k-1,S\setminus h)\geq \mu_i(k,S)$.
    We now proceed with our proof.
    If the first case holds, then we have that either $X_j= \emptyset$ or $X_j= \{g\}$.
    If $X_j = \emptyset$, then the argument holds obviously. In the case that $X_j= \{g\}$,
    by upper argument for $k=2$, $S=X_i\cup X_j$, and $h=g$, we get:
    \begin{equation*}
        v_i(X_i)=\mu_i\big(1,(X_i\cup X_j)\setminus\{g\}\big) 
        \geq \mu_i(2,X_i\cup X_j)
        \geq \frac{2}{3}\mu_i(2,X_i\cup X_j).    
    \end{equation*}
    So, assume the second case holds, then we have that $2v_i(X_i)\geq v_i(X_j)$, combining this
    with $v_i(X_i)\geq v_i(X_i)$ we get $v_i(X_i) \geq \frac{1}{3} v_i(X_i\cup X_j)$.
    Also, by Claim 3.1 of \cite{AMNS15}, for every agent $i$ and every subset
    of goods $S$ we have that $\mu_i(k,S) \leq \frac{v_i(S)}{k}$, so we get
    $v_i(X_i) \geq \frac{1}{3} v_i(X_i\cup X_j) \geq \frac{2}{3} \mu_i(2,X_i \cup X_j)$.
    So, the proof is complete.
\end{proof}

\section{Observations Regarding Valuations}
\label{valuation section}
    We say a valuation function $v$ is good cancelable 
    if for any four bundles $Q,R,S,\TT \subseteq M$ with $Q \cap S = R \cap \TT = \emptyset$,
    we have that $v(Q)\geq v(R)$ and $ v(S)>v(\TT)$ result in
    $v(Q \cup S) > v(R\cup \TT)$.
    
\begin{lemma}
    Every \goodcancelable\ valuation $v$ is also a restricted MMS-feasible valuation.
\end{lemma}

\begin{proof}
    Suppose that $v$ is a \goodcancelable\ valuation, 
    and $\X^k$ and $\Y^k$ are 
    two partitions of a subset of goods $S$. 
    If for any $\ell \in [k]$, we have $v(Y_\ell) \geq v(X_\ell)$, 
    then $\max(v(Y_1),\ldots,v(Y_k)) \geq v(Y_\ell) \geq v(X_\ell)\geq\min(v(X_1),\ldots,v(X_k))$. 
    So, assume that for every $\ell \in [k]$, we have $v(Y_\ell) < v(X_\ell)$. 
    Then, since the valuation function is \goodcancelable, by induction,
    we get that $v(\bigcup_{\ell=1}^k Y_\ell) < v(\bigcup_{\ell=1}^k X_\ell)$, which is a contradiction since 
    $\bigcup_{\ell=1}^k X_\ell = \bigcup_{\ell=1}^k Y_\ell $.
\end{proof}

\begin{definition}
    A valuation function $v$ is multiplicative  if for every $S\subseteq M$, we have:
    \begin{center}
        $v(S)= \Pi_{g\in S} v(g)$
    \end{center}
\end{definition}

\begin{corollary}
    Since additive and multiplicative valuations are good cancelable, they are restricted MMS-feasible, too.
\end{corollary}

\begin{definition}
    A valuation function $v$ is budget-additive if for every $S\subseteq M$, we have:
    \begin{center}
        $v(S)= \min( \sum_{g\in S} v(g), B)$
    \end{center}
    for some $B>0$.
\end{definition}

\begin{lemma}
    Every budget additive valuation function is restricted MMS-feasible.
\end{lemma}
\begin{proof}
    Suppose $S\subseteq M$ and $\X^k$ and $\Y^k$ are two partitions of $S$ to $k$ bundles.
    Since additive valuations are restricted MMS-feasible, we get that 
    $\max_{\ell\in [k]} \sum_{g\in X_\ell} v(g) \geq \min_{j \in [k]} \sum_{g \in Y_j} v(g)$;
    therefore, we get that:
    \begin{align*}
    \min\big(B,\max_{\ell\in [k]} \sum_{g\in X_\ell} v(g)\big) \geq \min \big(B, \min_{j \in [k]} \sum_{g \in Y_j} v(g) \big),        
    \end{align*}
    which in turn results in
    \begin{align*}
    \max_{\ell\in [k]} v(X_\ell) &= \max_{\ell\in [k]} \min(B, \sum_{g\in X_\ell} v(g)) \\
    &= \min\big(B,\max_{\ell\in [k]} \sum_{g\in X_\ell} v(g)\big) \\ 
    &\geq  \min \big(B, \min_{j \in [k]} \sum_{g \in Y_j} v(g) \big)  \\
    &= \min_{j \in [k]} \min (B, \sum_{g \in Y_j} v(g)) = \min_{j \in [k]} v(Y_j). 
    \end{align*}
\end{proof}

\begin{definition}
    A valuation function $v$ is unit demand if for every $S\subseteq M$, we have:
    \begin{center}
        $v(S)= \max_{g\in S} v(g)$
    \end{center}
\end{definition}

\begin{lemma}
    Every unit demand valuation function is restricted MMS-feasible.
\end{lemma}

\begin{proof}
       Suppose $S\subseteq M$ and $\X^k$ and $\Y^k$ are two partitions of $S$ to $k$ bundles.
       Since the maximum valued good is in at least one of the bundles, we get that:
       \begin{align*}
           \max_{\ell\in [k]} v(X_\ell) = \max_{g\in S} v(g) = \max_{\ell\in [k]} v(Y_\ell).
       \end{align*}
       Hence, we get the desired inequality:
      \begin{align*}
           \max_{\ell\in [k]} v(X_\ell) \geq  \min_{\ell\in [k]} v(Y_\ell).  
      \end{align*} 
\end{proof}

\section{Deferred Proofs from \cref{lemmas main body}}
\label{lemmas appendix}

    We now prove lemmas and observations of \cref{lemmas main body}.
    
\obsone*

\begin{proof}
    Suppose that \isbstffor{X_z}{i}{\X}, then if $X_\ell \lowerval{i} X_z \setminus x_z^i$, we get that agent $i$
    \efxenvies\ $X_z$ \with\ $X_\ell$, so \isnotefxffor{X_\ell}{i}{\X}. Also, if $X_\ell \greatereqval{i} X_z \setminus x_z^i$, since \isbstffor{X_z}{i}{\X},
    for every $X_j \in \X$ we get $X_\ell \greatereqval{i} X_z \setminus x_z^i \greatereqval{i} X_j \setminus x_j^i$.
    Hence, we get that if \isbstffor{X_z}{i}{\X}, then we have \isefxffor{X_\ell}{i}{\X} 
    if and only if  $X_\ell \greatereqval{i} X_z \setminus x_z^i$.
    If \isbstffor{X_z}{i}{\X}, since we have $X_z \greatereqval{i} X_z \setminus x_z^i$, we get that \isefxffor{X_z}{i}{\X}.

    Next, we show that if \isefxffor{X_z}{i}{\X}, then \iseflffor{X_z}{i}{\X}.
    Since \isefxffor{X_z}{i}{\X}, for every $X_j$ with more than one good, for any $g\in X_j$
    we have $X_j \setminus g \lowereqval{i} X_z$. Also, since for some $h \in X_j$ other than $g$ we have $g \in X_j\setminus h$,
    we get that $g \lowereqval{i} X_j \setminus h \lowereqval{i} X_z$. Therefore, $X_z$ is \eflf\ for $i$ in $\X$.

    Also, if \isefxffor{X_z}{i}{\X}, then by setting $\Y=\X$, we get that $X_z=Y_z$ is \efxf\ for $i$ in $\Y$, 
    and we have $Y_z\greatereqval{i}X_z$; therefore, $X_z$ is \eefxf\ and \mxsf\ for agent $i$ in $\X$.
    Also, if \iseefxffor{X_z}{i}{\X}, then there exists partition $\Y$ such that \isefxffor{Y_z}{i}{\Y} and we have $Y_z=X_z$, 
    so we have $Y_z \greatereqval{i} X_z$; hence, $X_z$ is \mxsf\ for agent $i$ in $\X$.
\end{proof} 

\obstwo*

\begin{proof}
    For any bundle $X_\ell' \in \X' \setminus \{ \PA \setminus \xii, X_r \cup \xii\}$, 
    we have $X_\ell' \in \X$, so
    since \isbstffor{\PA}{i}{\X}, we get that $v_i(\PA \setminus \xii) = \max_{g \in \PA} v_i(\PA \setminus g) \geq \max_{g \in X_\ell} v_i(X_\ell\setminus g)$. 
    Hence, agent $i$ does not \efxenvy\ 
    bundle $X_\ell$ \with\ bundle $\PA \setminus \xii$. 
    Therefore, agent $i$ does not \efxenvy\ bundle $X_\ell$ \with\ $\argmax_i(\PA \setminus \xii, X_r \cup \xii)$.
    Also, agent $i$ does not \efxenvy\ bundle $\argmin_i(\PA \setminus \xii, X_r \cup \xii)$ \with\ $\argmax_i(\PA \setminus \xii, X_r \cup \xii)$.
    Hence,  \isefxffor{\argmax_i(\PA \setminus \xii, X_r \cup \xii)}{i}{\xp}.
\end{proof}

\begin{restatable}{observation}{obsthree}
    \label{safe increase}
    Suppose $\X^k$ is a partition such that \iseflmffor{X_j}{i}{\X^k} and
    $X_\ell \greatereqval{i} X_j$, then \iseflmffor{X_\ell}{i}{\X^k}.
\end{restatable} 

\begin{proof}
    Since \ismxsffor{X_j}{i}{\X^k}, there exists some partition $\Z^k$ such that \isefxffor{Z_j}{i}{\Z^k} and $Z_j \lowereqval{i} X_j$. 
    Using the fact that $X_j \lowereqval{i} X_\ell$, we can infer that $Z_j\lowereqval{i} X_\ell$, 
    so \ismxsffor{X_\ell}{i}{\X^k}. 
    Also, because \iseflf{X_j}{i}{\X^k}, 
    for every bundle $X_b$ in partition $\X^k$, agent $i$ does not \eflenvy\ bundle $X_b$ \with\ bundle $X_j$,
    so either we have $|X_b|\leq 1$ or there exists some good $g$ in $X_b$ such that $X_j\greatereqval{i} g$ and $X_j \greatereqval{i} X_b \setminus g$.
    In both cases, agent $i$ does not \eflenvy\ bundle $X_b$ \with\ bundle $X_\ell$, in the first case by definition,
    and in the second case by the fact $X_\ell \greatereqval{i} X_j\greatereqval{i} g$ and $X_\ell \greatereqval{i} X_j \greatereqval{i} X_b \setminus g$.
\end{proof}

\begin{remark}
The Observation above is why our algorithm does not work for finding an EEFX+EFL allocation since 
a bundle with a value higher than an \eefxf\ bundle is not necessarily an \eefxf\ bundle.
\end{remark}

\eliminatecyclesend*

\begin{proof}
    After every execution of the while loop, there exists some agent 
    $i \in \range(f)$ such that her value for her associated bundle strictly increases
    w.r.t.\ her valuation function, 
    and for every agent  $j \in \range(f)$, her value for her associated bundle does not decrease.
    Therefore, this process cannot be repeated infinitely, 
    because there exist a finite number of assignments of bundles to agents.

    Next, it suffices to show that after the elimination of any cycle the invariants hold, 
    then, the proof will be held by induction.
    After the elimination of any cycle, 
    no two agents will be associated with a same bundle;
    therefore, $f$ remains an association function.
    After eliminating an envy cycle, every agent $i \in \range(f)$ 
    will either remain associated with her previous associated bundle  
    or to a bundle with a higher value w.r.t.\ her valuation function. 
    Hence, since the partition remains the same,
    and since $f$ was \mef\ before eliminating this cycle,
    by \cref{safe increase},  $f$ will still be \mef.
    Additionally, any free agent or free bundle remains free, and every unfree agent or bundle remains unfree.
    Hence, $D(f)$ and $\range(f)$ will not change.
\end{proof}


\chainexists*

\begin{proof}
    Let $c=(q_s,\ldots,q_1,\ell)$ be a subchain to $X_\ell$ with maximal length;
    then, we show that this is a chain.
    Suppose, on the contrary, that it is not. 
    Then, there exists some bundle $X_j$ with some associated agent $i=f(X_j)$ such that
    agent $i$ envies $X_{q_s}$ \with\ $X_j$. 
    Since there is no cycle in $G(\X,f)$, $j$ cannot be an index in $c$.
    Hence, $c'=(j,q_s,\ldots,\ell)$ is a subchain in $G(\X,f)$ 
    with a length greater than the length of $c$, which is a contradiction. 
    So, $c$ is a chain, and since a subchain with maximal length always exists, 
    the proof follows.
\end{proof}

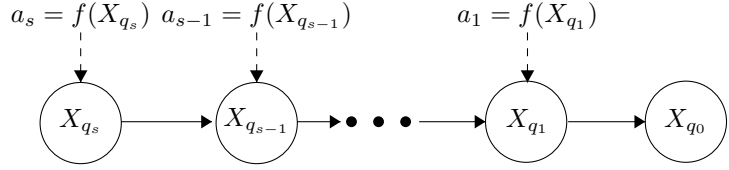
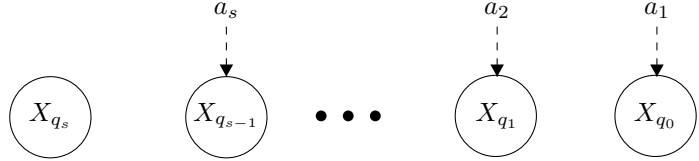
\begin{figure}
    \centering

     \begin{subfigure}[b]{1\textwidth}
         \centering
            \begin{tikzpicture}[scale=0.18]
            \tikzstyle{every node}+=[inner sep=0pt]
            \draw [black] (64.4,-16.5) circle (3);
            \draw (64.4,-16.5) node {$X_{q_0}$};
            \draw [black] (52.6,-16.5) circle (3);
            \draw (52.7,-16.5) node {$X_{q_1}$};
            \draw [black, dashed] (52.7,-9.9) -- (52.7,-13.6);
            \fill [black] (52.7,-13.6) -- (52.2,-12.7) -- (53.2,-12.7);
            \node[draw=none,align=left] at (52.7,-8.7) {$a_1=f(X_{q_1})$};
            \fill [black] (43.7,-16.5) circle (0.4);
            \fill [black] (41.7,-16.5) circle (0.4);
            \fill [black] (39.7,-16.5) circle (0.4);
            \draw [black] (32.7,-16.6) circle (3);
            \draw (32.7,-16.5) node {$X_{q_{s-1}}$};
            \draw [black, dashed] (32.7,-9.9) -- (32.7,-13.6);
            \fill [black] (32.7,-13.6) -- (32.2,-12.7) -- (33.2,-12.7);
            \node[draw=none,align=left] at (32.7,-8.7) {$a_{s-1}=f(X_{q_{s-1}})$};

            \draw [black] (19.7,-16.6) circle (3);
            \draw (19.7,-16.5) node {$X_{q_s}$};
            \draw [black, dashed] (19.7,-9.9) -- (19.7,-13.6);
            \fill [black] (19.7,-13.6) -- (19.2,-12.7) -- (20.2,-12.7);
            \node[draw=none,align=left] at (19.7,-8.7) {$a_s=f(X_{q_s})$};

            \draw [black] (35.7,-16.5) -- (38.9,-16.5);
            \fill [black] (38.9,-16.5) -- (38.1,-16) -- (38.1,-17);

            \draw [black] (22.7,-16.5) -- (29.4,-16.5);
            \fill [black] (29.4,-16.5) -- (28.6,-16) -- (28.6,-17);

            \draw [black] (55.7,-16.5) -- (61.4,-16.5);
            \fill [black] (61.4,-16.5) -- (60.6,-16) -- (60.6,-17);
            \draw [black] (44.7,-16.5) -- (49.6,-16.5);
            \fill [black] (49.6,-16.5) -- (48.8,-16) -- (48.8,-17);
            \end{tikzpicture}
            
         \caption{This figure shows subchain $c$. Vertices correspond to bundles in a subchain, 
         and non-dashed arrows correspond to envy edges along the subchain.}
         \label{fig:c def}
     \end{subfigure}

\medskip
     \begin{subfigure}[b]{1\textwidth}
         \centering
            \begin{tikzpicture}[scale=0.18]
            \tikzstyle{every node}+=[inner sep=0pt]
            \draw [black] (64.4,-16.5) circle (3);
            \draw (64.4,-16.5) node {$X_{q_0}$};
            \draw [black, dashed] (64.5,-9.9) -- (64.5,-13.6);
            \fill [black] (64.5,-13.6) -- (64,-12.7) -- (65,-12.7);
            \node[draw=none,align=left] at (64.5,-8.7) {$a_1$};
            
            \draw [black] (52.6,-16.5) circle (3);
            \draw (52.7,-16.5) node {$X_{q_1}$};
            \draw [black, dashed] (52.7,-9.9) -- (52.7,-13.6);
            \fill [black] (52.7,-13.6) -- (52.2,-12.7) -- (53.2,-12.7);
            \node[draw=none,align=left] at (52.7,-8.7) {$a_2$};
            \fill [black] (43.7,-16.5) circle (0.4);
            \fill [black] (41.7,-16.5) circle (0.4);
            \fill [black] (39.7,-16.5) circle (0.4);
            \draw [black] (32.7,-16.6) circle (3);
            \draw (32.7,-16.5) node {$X_{q_{s-1}}$};
            \draw [black, dashed] (32.7,-9.9) -- (32.7,-13.6);
            \fill [black] (32.7,-13.6) -- (32.2,-12.7) -- (33.2,-12.7);
            \node[draw=none,align=left] at (32.7,-8.7) {$a_{s}$};

            \draw [black] (19.7,-16.6) circle (3);
            \draw (19.7,-16.5) node {$X_{q_s}$};

            \end{tikzpicture}
         \caption{This figure shows agents' association with bundles after executing \textsc{ShiftSubChain} along the
         subchain in \cref{fig:c def}.
         As it is depicted, $X_{q_s}$ will be free after this process,
         and if $s>0$, $X_{q_0}$ will have an associated agent.}
         \label{fig:c in X}
     \end{subfigure}

    \caption{This figure demonstrates how association function changes after \textsc{ShiftSubChain}. Dashed arrows show agents associated with bundles.}
    \label{fig:shift-subchain}
\end{figure}

%


\lemone*
\begin{proof}
    During this process, no agent from outside of $\range(f)$ gets associated with any bundle,
    so we get that $\range(f')\subseteq \range(f)$. Also, in the process, whenever
    we associate some agent with some bundle, we take her previously associated bundle away from her,
    so in the end, every agent will be associated with at most one bundle.
    So, $f'$ will be an association function for $\X$.
    Also, since for every $i\in \range(f')$ we have that either she has her previous associated bundle
    or the next bundle in the subchain (by definition of subchain), we get
    that her new associated bundle does not have a lower value w.r.t.\ her valuation,
    so by \cref{safe increase}, since her previous bundle was \eflmf, 
    her new bundle is \eflmf\ too.

    Also, since after updating $f$ along $c$, every agent in $\range(f')$
    will be associated with a bundle with a value not lower than the bundle she was associated with before 
    (w.r.t.\ her valuation function), if an agent did not envy bundle $\SC$ \with\ her previous associated bundle,
    she will not envy bundle $\SC$ \with\ her new associated bundle too. 
    So, if $\SC$ was a source before the update, it would be a source after the update, too.
    If length of $c$ is zero, then $q=p$, and shift chain only sets, $f(q)=0$, so we get $D(f')=D(f)\setminus \{q\} = D(f) \cup \{q\} \setminus \{q\}=D(f) \cup \{q\} \setminus \{p\}$.
    Otherwise, index $q$ gets an associated agent in the for loop of \textsc{ShiftSubChain},
    and other indices that get a new associated agent had an associated agent by the definition of subchain.
    Also, the only index that loses an associated agent is $p$; hence,
    $D(f')= D(f) \cup \{q\} \setminus \{p\}$ (see \cref{fig:shift-subchain}).
\end{proof}

\noindent Recall that for some agent $u$ and some bundle $X_z$,
we define $x_z^u \in \arg\max_{g \in X_z}v_u(X_z\setminus g)$.


\lemoneEFLtwo*

\begin{proof}
    Note that in Case 1, the new bundle $X'_\ell$ satisfies $X'_\ell \greatereqval{i} X_\ell$ since $X_\ell$ did not lose any goods, and $X_\ell \greatereqval{i} \SC$ by definition of $\SC$. Therefore, $X'_\ell \greatereqval{i} \SC$. On the other hand, in Case 2 we have $X'_\ell = \PA' = \PA \setminus \xii$, 
    since $p=\ell$ and $\PA$ lost good $\xii$. Also, by definition of this case, we have  $\PA' = \PA \setminus \xii \greatereqval{i} \SC$, 
    so, once again, $X'_\ell \greatereqval{i} \SC$. 
    
    Then, in Case 1, we show that for every $g \in \PA$, we have that $X_\ell \greatereqval{i} g$.
    Since \iseflffor{X_\ell}{i}{\X} and $\PA$ has at least two goods,  there exists some good $h\in \PA$ such that
    $X_\ell \greatereqval{i} h$, $X_\ell \greatereqval{i} \PA \setminus h$. If $g = h$, then the argument follows from $X_\ell \greatereqval{i} h$, and if $g \ne h$,
    then we have $\{g\} \subseteq  \PA \setminus h$, 
    so since the valuations are monotone, we get the argument from $g \lowereqval{i} \PA \setminus h \lowereqval{i} X_\ell$.
    
    In the case that $\ell \ne p$, we have $X'_\ell \greatereqval{i} \SC$, $X'_\ell \greatereqval{i} g$; therefore, 
    agent $i$ does not \eflenvy\ bundle $\SC'= \SC \cup g$ \with\ bundle $X'_\ell$, so  since other bundles did not get any new goods
    and \iseflffor{X_\ell}{i}{\X},
    we get $X'_\ell$ is \eflf\ in the new partition for agent $i$.
    Also, since \ismxsffor{X_\ell}{i}{\X}, and $X'_\ell \greatereqval{i} X_\ell$ (by $\ell \ne p)$, we get \ismxsffor{X'_\ell}{i}{\xp}.
    
    In the case that $\ell = p$, since $\PA$ has at least two goods, 
    we have that $X'_\ell = \PA \setminus \xii \greatereqval{i} \xii$; therefore, agent $i$ does not \eflenvy\ bundle $\SC \cup \xii$ \with\ bundle $X'_\ell$ . 
    Also, since $\PA$ was \bstf\ for agent $i$ in $\X$, 
    agent $i$ does not \efxenvy\ any bundle
    in $\X$ \with\ bundle $X'_\ell = \PA \setminus \xii$; therefore, \iseflffor{X'_p}{i}{\xp}.
\end{proof}

\section{Deferred Proofs of \cref{lemmas main body}}
\label{chain dominance appendix}

\begin{lemma}
\label{dominance relation}
    Suppose that $\X, \xp, \xz$ are partitions with $k$ bundles with association functions
    $f,f',f''$, respectively,
    $c''\in \ckxz$, 
    $c'\in \ckxp$, 
    and $c\in \ckx$. 
    Then, if $c''\ggeq c'$ and $c'\ggeq c$, then $c''\ggeq c$.
    Also, if at least one of these dominations is strict, then $c''\gg c$.
\end{lemma}

\begin{proof}
    If $c=c'$ or $c'=c''$, the arguments obviously hold, 
    so assume that  $c\ne  c'$ and $c' \ne c''$, so dominations are strict.
    Let $c = (q_{s}, \ldots,q_1, k)$, 
        $c' = (q'_{s'}, \ldots,q'_1, k)$, and
        $c''= (q''_{s''}, \ldots,q''_1, k)$.

    \noindent Since $c' \gg c$, we have $X_{k} \subseteq X'_{k}$, and one of the following hold:
    \begin{itemize}
        \item \textbf{A1)} $s'<s$, and for every $\ell \in [s']:$ 
                            $X'_{q'_\ell} = X_{q_\ell}$ and $f'(q'_\ell) = f(q_\ell)$.
        \item \textbf{A2)} There exists some $b \in [s]$ such that for every $\ell \in [b-1]$:
                            $X'_{q'_\ell} = X_{q_\ell}$ and $f'(q'_\ell) = f(q_\ell)$.
                            Also, we have  $f'(q'_b) = f(q_b)$, and at least one of the 
                            $X'_{q'_b} \greaterval{f(q_b)} X_{q_b}$ or $X_{q_b} \subset X'_{q'_b}$ holds.
        \item \textbf{A3)} $X_{k} \subset X'_{k}$.
    \end{itemize}
    
    \noindent Since $c'' \gg c$, so $X'_{k} \subseteq X''_{k}$ and one of the following holds:
    \begin{itemize}
        \item \textbf{B1)} $s''<s'$, and for every $\ell \in [s'']:$ 
                           $X''_{q''_\ell} = X'_{q'_\ell}$ and $f''(q''_\ell) = f'(q'_\ell)$
        \item \textbf{B2)} There exists some $b' \in [s']$ such that for every $\ell \in [b'-1]$: 
                           $X''_{q''_\ell} = X'_{q'_\ell}$ and $f''(q''_\ell) = f'(q'_\ell)$.
                           Also, $f''(q''_{b'}) = f'(q'_{b'})$, and at least one of the 
                           $X''_{q''_{b'}} \greaterval{f'(q'_{b'})}  X'_{q'_{b'}}$ or 
                           $X'_{q'_{b'}} \subset X''_{q''_{b'}}$ holds.
        \item \textbf{B3)} $X'_{k} \subset X''_{k}$.
    \end{itemize}    
    
    From $X_{k} \subseteq X'_{k}$ and   $X'_{k} \subseteq X''_{k}$,  we get that  $X_{k} \subseteq X''_{k}$. So, we need to show that one of the three cases in \cref{chain dominance}
    holds.
    If any of the cases A3 or B3 holds, 
    we get that either  $X_{k} \subset X'_{k}$ or $X'_{k} \subset X''_{k}$, 
    which any of those results in $X_{k}\subset X''_{k}$, 
    and Case 3 in $\cref{chain dominance}$ holds, so assume neither occurs.
    Next, we consider the following four cases based on which of the stated cases hold.

    \noindent\textbf{$\bullet$ \CaseA. A1 and B1 hold:}
    From $s''<s'$ and $s'<s$, we get $s''<s$. 
    We have that for every $\ell \in [s'']$, $f''(q''_\ell) = f'(q'_\ell)$ and $X''_{q''_\ell} =  X'_{q'_\ell}$.
    Since $s''<s$,  we get that $\ell \in [s]$, 
    so we get  $f'(q'_\ell)= f(q_\ell)$ and $X'_{q'_\ell} =  X_{q_\ell}$.
    Hence, we get that for every $\ell \in [s'']$, $f''(q''_\ell) = f(q_\ell)$ and   
    $X''_{q''_\ell} = X_{q_\ell}$.
    Therefore, $c'' \gg c$ in this case by case 1 of \cref{chain dominance}.

    \noindent\textbf{$\bullet$ \CaseB. A1 and B2 hold:}
    Then by B2,
    there exists  $b' \in [s']$ such that 
    for every $\ell \in [b'-1]$ we have 
    $X''_{q''_\ell} = X'_{q'_\ell}$ and $f''(q''_\ell) = f'(q'_\ell)$.
    By A1, we get $s'<s$; therefore, $b'<s$, so
    $b' \in [s]$, and for every $\ell \in [b']$ we get that
    $f'(q'_\ell) = f(q_\ell)$ and $X'_{q'_\ell} =  X_{q_\ell}$.
    Hence, for every $\ell \in [b'-1]$, 
    we have $f''(q''_\ell) = f(q_\ell)$ and $X''_{q''_{\ell}} =  X_{q_\ell}$.
    By B2, $f''(q''_{b'}) = f'(q'_{b'})$, so by $f'(q'_{b'}) = f(q_{b'})$,
    we get $f''(q''_{b'}) = f(q_{b'})$.
    Also, by B2, we have that either 
    $X''_{q''_{b'}} \greaterval{f'(q'_{b'})}  X'_{q'_{b'}}$ or 
    $X'_{q'_{b'}} \subset X''_{q''_{b'}}$, where in the first case,
    since $X'_{q'_{b'}} =  X_{q_{b'}}$ and $f'(q'_{b'}) = f(q_{b'})$,
    we get  $X''_{q''_{b'}} \greaterval{f(q_b)}  X_{q_{b}}$,
    and in the second case,  
    we get $X_{q_{b'}}=X'_{q'_{b'}} \subset X''_{q''_{b'}}$.

    \noindent\textbf{$\bullet$ \CaseC. A2 and B1 hold:}
    Then there exists some $b \in [s]$ such that for every $\ell \in [b-1]$ we have 
    $f'(q'_\ell) = f(q_\ell)$, $X'_{q'_\ell} = X_{q_\ell}$.
    Also, for every $\ell \in [s'']$ we have 
    $f''(q''_\ell) = f'(q'_\ell)$ and $X''_{q''_\ell} = X'_{q'_\ell}$.
    Hence, for every $\ell \in [\min(s'', b-1)]$, we have 
    $f''(q''_\ell) = f'(q'_\ell)= f(q_\ell)$, $X''_{q''_\ell} = X'_{q'_\ell} =  X_{q_\ell}$.
    So, if $s'' < b$, since $b\leq s$, we get $s''<s$,
    and for every $\ell \in [s'']$,  
    $f''(q''_\ell) = f(q_\ell)$, $X''_{q''_\ell} = X_{q_\ell}$. 
    So, $c'' \gg c$ by case one of \cref{chain dominance}.
    
    So, assume that $s''\geq b$. Then, we get that for every $\ell \in [b-1]$, 
    $f''(q''_\ell) = f(q_\ell)$, $X''_{q''_\ell} = X_{q_\ell}$. 
    Also, since B1 holds, 
    we have that $X''_{q''_b} = X'_{q'_b}$ and $f''(q''_b) = f'(q'_b)$.
    Hence, $f(q_b) = f'(q'_b) = f''(q''_b)$. 
    Also, we have that either 
    $X'_{q'_b} \greaterval{f(q_b)} X_{q_b}$ or $X_{q_b} \subset X'_{q'_b}$, 
    where in the first case, we get 
    $X''_{q''_b} = X'_{q'_b} \greaterval{f(q_b)} X_{q_b}$, 
    and in the second case, we get 
    $X_{q_b} \subset X'_{q'_b} = X''_{q''_b} $.
    Therefore, by Case 2 of \cref{chain dominance}, $c'' \gg c$.

    \noindent\textbf{$\bullet$ \CaseD. A2 and B2 hold:}
    let $b'' = \min(b, b')$, then for every $\ell \in [b''-1]$, we have 
    $f''(q''_\ell) =f'(q'_\ell) = f(q_\ell)$, 
    $X''_{q''_\ell}=X'_{q'_\ell} = X_{q_\ell}$.
    We have that $f''(q''_{b''})= f'(q'_{b''}) = f(q_{b''})$.
    Since valuations are monotone, we have that 
    $X''_{q''_{b''}} \greatereqval{f(q_{b''})} X'_{q'_{b''}} \greatereqval{f(q_{b''})} X_{q_{b''}}$,
    so if one of the inequalities is strict, 
    we get $X''_{q''_{b''}}  \greatereqval{f(q_{b''})} X_{q_{b''}}$, and we are done.
    So assume not, then
    $X_{q_{b''}} \subseteq  X'_{q'_{b''}} \subseteq X''_{q''_{b''}} $, and
    since $b''$ equals with at least one of the $b$ and $b'$, 
    we get that one of the inclusions is strict, so we get 
    $X_{q_{b''}} \subset X''_{q''_{b''}} $. 
    Hence, $c'' \gg c$ by case 2 of \cref{chain dominance}.
    So, the proof is complete.
\end{proof}

\begin{observation}
\label{not stric makes equal}
    If $(\xp,f')\ggeq (\X,f)$, but $(\xp,f')$ does not strictly dominates $(\X,f)$, then $\ckxp = \ckx$. 
\end{observation}
\begin{proof}
    $(\xp,f')$ does not strictly dominate $(\X,f)$ and  $(\xp,f')\ggeq (\X,f)$,
    so for every $c' \in \ckxp$, we have $c \in \ckx$; hence, $\ckxp \subseteq \ckx$.
    Since domination is not strict, we get that $C_k(\xp,f')= C_k(\X,f).$
\end{proof}

\begin{observation}
\label{strict removes}
    If  chain $c'$ in $\ckxp$ \emph{strictly dominates} chain $c$ in $\ckx$, 
    then $c \notin \ckxp$.
    Also, $(\X,f)\ne (\xp,f')$.
\end{observation}
\begin{proof}
    Let $c'= (q'_{s'},\ldots, k)$ and $c=(q_{s},\ldots,k)$.
    Since $c' \gg c$,   
    either condition 1, 2, or 3 in \cref{chain dominance} holds.
    If condition 1 holds, 
    then $s'<s$, so $c'$ is a sub-chain of $c$ with lower length,
    so if $c$ is a chain in $\ckxp$, then $c'$ will not be a chain, which is a contradiction, 
    so $c\notin \ckxp$.
    If condition 2 holds, then we get for some $b\in [s]$, $X'_{q'_b} \ne  X_{q_b}$ and $f'(q'_b)=f(q_b)$; 
    therefore, the bundle associated with agent $f'(q'_b)$ 
    is different in the two partitions, so $c\notin \ckxp$.
    If condition 3 holds, then we have that $X_{k} \ne X'_{k}$, $c\notin \ckxp$..
    Therefore, $c\notin \ckxp$ in any case, and $(\X,f)\ne (\xp,f')$.
\end{proof}


\transitivityXf*
\begin{proof}
    Since $(\xz,f'') \ggeq (\xp,f')$, for every chain $c''\in \ckxz$, there exists some chain $c'\in \ckxp$ such that $c''\gg c'$,
    and since $(\xp,f')\ggeq (\X,f)$, there exists some chain $c\in \ckx$ such that $c'\gg c$; 
    therefore, by \cref{dominance relation}, $c'' \gg c$. 
    Hence, for every chain $c''\in \ckxz$,  there exists some chain $c\in \ckx$ such that $c''\gg c$. 
    Hence, $(\xz,f'')\ggeq (\X,f)$.
    Now suppose that one of the dominations is strict.
    If $(\xz,f'')$ does not strictly dominate $(\X,f)$,  because  $(\xz,f'')\ggeq (\X,f)$, by \cref{not stric makes equal}, 
    we get that $C_k(\xz,f'')= C_k(\X,f).$
    
    Next, based on the definition of strict domination of sets of chains, 
    one of the following cases holds.

    \noindent\textbf{$\bullet$ \CaseA.}
    \textbf{$(\xz,f'') \gg (\xp,f')$:}


    \noindent\textbf{$\bullet$ Subcase 1.}
    \textbf{There exists some chain $c''\in \ckxz$ and some chain $c'\in \ckxp$ such that $c''\gg c'$:}
    Then, since $(\xp,f')\gg (\X,f)$, there exists some $c\in \ckx$ such that $c'\ggeq c$; 
    therefore, by \cref{dominance relation}, $c''\gg c$.

    \noindent\textbf{$\bullet$ Subcase 2.} 
    \textbf{$\ckxz \subset \ckxp$:}
    Then, there exists a $c'\in \ckxp$ such that $c' \notin \ckxz$. 
    Since $(\xp,f') \ggeq (\X,f)$, there exists a $c\in \ckx$ such that $c'\ggeq c$.
    If $c'=c$, then we get $c \notin \ckxz$.
    Otherwise, $c'\gg c$, so by \cref{strict removes}, $c \notin \ckxp$, and by $\ckxz\subset \ckxp$, we get $c\notin \ckxz$.
    Hence, in each case $c\notin \ckxz$, which means $\ckx\ne \ckxz$, so $(\xz,f'') \gg (\X,f)$.

    \noindent\textbf{$\bullet$ \CaseB.}
    \textbf{$(\xp,f') \gg (\X,f)$ and $(\xz,f'')$ does not strictly dominate $(\xp,f')$:}
    Then, $\ckxz=\ckxp$. If $\ckxp \subset \ckx$, we get $\ckxz \subset \ckx$.
    If there exists a chain $c'\in \ckxp$ and a chain $c\in \ckx$ such that $c'\gg c$, then by \cref{strict removes},
    $c \notin \ckxp$, so $c \notin \ckxz$, 
    which means $C_k(\X,f)\ne C_k(\xz,f'')$, so $(\xz,f'') \gg (\X,f)$. 
\end{proof}

\dominancechange*
\begin{proof}
    Since $(\xp,f')\gg (\X,f)$, either we have $C_k(\xp,f') \subset C_k(\X,f)$, 
    which in this case we are done, or there exists a chain $c'\in \ckxp$
    and a chain $c \in \ckx$ such that 
    $c'\gg c$.
    Hence, by \cref{strict removes}, we get that $c \notin C_k(\xp,f')$, so $C_k(\X,f) \ne C_k(\xp,f')$, and the proof is complete.
\end{proof}

\end{document}